\documentclass[12pt]{article}

\usepackage[margin=1in]{geometry}

\usepackage{amsmath, amssymb, amsthm, amssymb}
\usepackage{amsfonts}
\usepackage{graphicx}
\usepackage[frame,ps,matrix,arrow,curve,rotate]{xy}
\usepackage{enumerate}
\usepackage{hyperref}
\usepackage{siunitx}

\usepackage{blkarray}

\newtheorem{theorem}{Theorem}[section]
\newtheorem{lemma}[theorem]{Lemma}

\newtheorem{conjecture}[theorem]{Conjecture}
\newtheorem{corollary}[theorem]{Corollary}

\theoremstyle{definition}
\newtheorem{definition}[theorem]{Definition}

\theoremstyle{remark}

\newcommand{\mc}[1]{\mathcal{#1}}

\newcommand{\bfmu}{\boldsymbol{\mu}}
\newcommand{\bfSigma}{\boldsymbol{\Sigma}}
\newcommand{\bfGamma}{\boldsymbol{\Gamma}}
\newcommand{\bfX}{\mathbf{X}}
\newcommand{\bfY}{\mathbf{Y}}
\newcommand{\bfS}{\mathbf{S}}
\newcommand{\bfx}{\mathbf{x}}
\newcommand{\bfy}{\mathbf{y}}
\newcommand{\bfz}{\mathbf{z}}

\DeclareMathOperator{\flow}{flow}

\DeclareMathOperator{\Cov}{Cov}
\DeclareMathOperator{\Var}{Var}
\DeclareMathOperator{\Margin}{Margin}
\DeclareMathOperator{\lin}{lin}

\makeatletter
\newcommand\mathcircled[1]{%
  \mathpalette\@mathcircled{#1}%
}
\newcommand\@mathcircled[2]{%
  \tikz[baseline=(math.base)] \node[draw,circle,inner sep=1pt] (math) {$\m@th#1#2$};%
}
\makeatother

\begin{document}

\title{An Analysis of Random Elections\\ with Large Numbers of Voters\thanks{We thank Wesley Holliday and Eric Pacuit for extensive comments and discussions.}}
\author{Matthew Harrison-Trainor}

\maketitle

\begin{abstract}
	In an election in which each voter ranks all of the candidates, we consider the head-to-head results between each pair of candidates and form a labeled directed graph, called the margin graph, which contains the margin of victory of each candidate over each of the other candidates. A central issue in developing voting methods is that there can be cycles in this graph, where candidate $\mathsf{A}$ defeats candidate $\mathsf{B}$, $\mathsf{B}$ defeats $\mathsf{C}$, and $\mathsf{C}$ defeats $\mathsf{A}$. In this paper we apply the central limit theorem, graph homology, and linear algebra to analyze how likely such situations are to occur for large numbers of voters. There is a large literature on analyzing the probability of having a majority winner; our analysis is more fine-grained. The result of our analysis is that in elections with the number of voters going to infinity, margin graphs that are more cyclic in a certain precise sense are less likely to occur.
\end{abstract}

\section{Introduction}

The Condorcet paradox is a situation in social choice theory where every candidate in an election with three or more alternatives would lose, in a head-to-head election, to some other candidate. For example, suppose that in an election with three candidates $\mathsf{A}$, $\mathsf{B}$, and $\mathsf{C}$, and three voters, the voter's preferences are as follows:
\begin{center}
\begin{tabular}{c|c|c|c}
	& First choice & Second choice & Third choice \\\hline
	Voter 1 & $\mathsf{A}$ & $\mathsf{B}$ & $\mathsf{C}$ \\\hline
	Voter 2 & $\mathsf{B}$ & $\mathsf{C}$ & $\mathsf{A}$ \\\hline
	Voter 3 & $\mathsf{C}$ & $\mathsf{A}$ & $\mathsf{B}$ 
\end{tabular}
\end{center}
There is no clear winner, for one can argue that $\mathsf{A}$ cannot win as two of the three voters prefer $\mathsf{C}$ to $\mathsf{A}$, that $\mathsf{B}$ cannot win as two of the three voters prefer $\mathsf{A}$ to $\mathsf{B}$, and that $\mathsf{C}$ cannot win as two of the three voters prefer $\mathsf{B}$ to $\mathsf{C}$.

More formally, fix a set $\mathcal{V} = \{ \mathsf{v}_1,\ldots,\mathsf{v}_n \}$ of voters and a set $\mc{C} = \{ \mathsf{c}_1,\ldots,\mathsf{c}_\ell\}$ of candidates. Let $\mc{L} = \mc{L}(\mc{C})$ be the set of all linear orders on $\mc{C}$; we think of such a linear order as a ranking of the candidates. A \textit{profile} is a map $P \colon \mc{V} \to \mc{L}$, mapping each voter $\mathsf{v}$ to a ranking $P(\mathsf{v})$ of the candidates; we call $P(\mathsf{v})$ voter $\mathsf{v}$'s \textit{ballot}. So a profile is exactly the data we might get from an election. We write $\mathsf{c} >^P_{\mathsf{v}} \mathsf{d}$ if voter $\mathsf{v}$ prefers candidate $\mathsf{c}$ to $\mathsf{d}$ in the profile $P$.

Given a profile $P$, the \textit{margin} of one candidate $\mathsf{c}$ over another $\mathsf{d}$ is the margin of victory/loss of $\mathsf{c}$ over $\mathsf{d}$ in a direct comparison:
\[ \Margin_P(\mathsf{c},\mathsf{d}) = \#\{ \mathsf{v} \in \mc{V} : \mathsf{c} >^P_\mathsf{v} \mathsf{d} \} - \#\{ \mathsf{v} \in \mc{V} : \mathsf{d} >^P_\mathsf{v} \mathsf{c} \}.\]
If $\Margin_P(\mathsf{c},\mathsf{d}) > 0$ we say $\mathsf{c}$ is \textit{majority preferred} to $\mathsf{d}$. We will always consider the number of voters to be odd so that for every pair of candidates, one is majority preferred to the other. We can construct a labeled directed \textit{margin graph} $\mc{M}(P)$ whose vertices are the candidates, and with an edge from $\mathsf{c}$ to $\mathsf{d}$, labeled with $\Margin_P(\mathsf{c},\mathsf{d})$, exactly when $\mathsf{c}$ is majority preferred to $\mathsf{d}$. For example, the margin graph of the profile described above is
\[\xymatrix{&\mathsf{A}\ar[dl]_1 \\\mathsf{B}\ar[rr]_1&&\mathsf{C}\ar[ul]_1}\]
When we forget the margins of victory, we obtain a tournament which we call the \textit{majority graph} $\mc{G}(P)$ of $P$. 

\begin{theorem}[Debord \cite{Debord}]
	For any labeled tournament $\mc{M}$, such that all weights of edges have the same parity, there is a profile $P$ such that $\mc{M}$ is the margin graph of $P$.
\end{theorem}

The Condorcet paradox occurs when there is a cycle in the margin graph of a profile, so that there are candidates $\mathsf{c}_1,\ldots,\mathsf{c}_k$ such that $\mathsf{c}_1$ is majority preferred to $\mathsf{c}_2$, $\mathsf{c}_2$ to $\mathsf{c}_3$, and so on, and $\mathsf{c}_k$ is majority preferred to $\mathsf{c}_1$. If there is no cycle in the margin graph of a profile $P$, then the margin graph is just a linear ordering of the candidates, and it is plausible that the winner should be the greatest candidate according to this ordering. A central problem of voting theory is to come up with a voting method---formally a function mapping each profile $P$ to a set of winning candidates (or sometimes to a ranking of all the candidates)---which deals as well as possible with cycles in the margin graph.

Thus an important area of research has been to identify how often there occur cycles in the margin graph, both in historical situations and theoretically under various assumptions on the voters. Riker \cite{Riker} argues that various amendments in the  Agricultural Appropriation Act of 1953 in the US House of Representatives formed a cycle; Bjurulf and Niemi \cite{BjurulfNiemi} found similar situations in the Swedish parliament; Stensholt \cite{Stensholt} found a cycle in a decision by the Norwegian national assembly; see also Van Deemen \cite{vD}, Kurrild-Klitgaard \cite{Kurrild-Klitgaard}, and Truchon \cite{Truchon}.

In this paper we take the more theoretical point of view where we assume that voters fill out their ballot randomly according to some probability distribution, and we consider the probability of a paradox occurring There is also a large literature of results here. The book by Gehrlein \cite{Gehrlein06} is an excellent reference for what is known.

To begin, we must make an assumption about the probability distribution of ballots for each voter. We assume for the rest of the paper that each voter is equally likely to pick any of the $\ell!$ linear orders on the $\ell$ candidates, an assumption referred to in the literature as the Impartial Culture (IC) condition.\footnote{We note that there is a wide range of work on what happens under assumptions other than the Impartial Culture condition. One can apply the same sort of analysis in this paper to these other assumptions, though the covariance matrix $\bfSigma$ obtained will be different. It would be very interesting to compare the results obtained. We leave this for future work.}

An important probability is the probability $P_W(n,\ell)$ of having a Condorcet winner---a single candidate who is majority preferred to each other candidate. Such a candidate can be argued to be a clear winner, and an important class of voting methods, the Condorcet methods, select such a candidate as the winner. Using increasingly sophisticated methods, Sevcik \cite{Sevcik}, DeMeyer and Plott \cite{DeMeyerPlott}, Niemi and Weisberg \cite{NiemiWeisberg}, Garman and Kamien, \cite{GarmanKamien}, and Gehrlein and Fishburn \cite{GehrleinFishburn76,GehrleinFishburn79a,Gehrlein99b} calculated these probabilities for small numbers of voters and candidates. These values are included in Table \ref{table1}. One can also calculate other related probabilities, such as the probability $P_T(n,\ell)$ of having a transitive majority graph. Gehrlein \cite{Gehrlein1988,Gehrlein1989} calculated these values for various $n$ and $\ell$.

As the number of voters becomes large, it is known that due to the central limit theorem, the distribution of margin graphs approaches a multivariate normal distribution. Guilbaud \cite{Guilbaud52} was the first to compute for three candidates the probability $P_W(\infty,3) = \lim_{n \to \infty} P_W(n,3) = 0.9123$ of having a majority winner, though he did not use the central limit theorem. Niemi and Weisberg \cite{NiemiWeisberg} and Garman and Kamien \cite{GarmanKamien} noted the limiting behaviour to a multivariate normal distribution and used this to calculate values of $P_W(\infty,\ell) = \lim_{n \to \infty} P_W(n,\ell)$ for various numbers $\ell$ of candidates. These probabilities are shown in Table \ref{table1} as well. Gehrlein and Fishburn \cite{GehrleinFishburn78} computed other probabilities such as that of being transitive or having a Hamiltonian cycle. Most of these are numerical approximations.

\begin{table}\label{table1}
	\caption{Probability of a majority winner for $n$ voters and $\ell$ candidates}
	\medskip{}
\centering

\begin{tabular}{ccccccccc}
	 &  &&& $\ell$ \\ \cline{2-9}
	$n$ & 3 & 4 & 5 & 6 & 7 & 8 & 9 & 10 \\ \hline
	3 & 0.94444 & 0.88889 & 0.84000 & 0.79778 & 0.76120 & 0.72925 \\
	5 & 0.93056 & 0.86111 & 0.80048 & 0.74865 & 0.70424 & 0.66588 \\
	7 & 0.92498 & 0.84997 & 0.78467 & 0.72908 & 0.68168 & 0.64090 \\
	9 & 0.92202 & 0.84405 & 0.77628 & 0.71873 & 0.66976 \\
	11 & 0.92019 & 0.84037 & 0.77108 & 0.71231 & 0.66238 \\
	13 & 0.91893 & 0.83786 & 0.76753 & 0.70194 & 0.65736 \\
	15 & 0.91802 & 0.83604 & 0.76496 & 0.70476 & 0.65372 \\
	\vdots & \vdots & \vdots & \vdots & \vdots & \vdots & \vdots & \vdots & \vdots \\
	$\infty$ & 0.9123 & 0.8245 & 0.7487 & 0.6848 & 0.6308 & 0.5849 & 0.5455 & 0.5113 
\end{tabular}
\end{table}
We quickly notice that the probability of avoiding a paradox is higher than one might naively expect. Given three candidates, there are eight possible majority graphs: six linear orders and two cycles. If these all occurred with equal likelihood, then we would expect the probability $P_W(\infty,3) = 3/4$ of having a Condorcent winner. In fact $P_W(\infty,3)$ is much higher than this. Intuitively, this is because each voter's ballot is linearly ordered, and this makes the majority graph of the whole profile more likely to be linearly ordered. More formally, one can compute that the event that $\mathsf{A}$ beats $\mathsf{B}$ is positively correlated with the event that $\mathsf{A}$ beats $\mathsf{C}$.

\medskip{}

So far most research has been on the probability of particular events occurring, such as having a Condorcet winner, being transitive, having a Hamiltonian cycle, etc. In this paper we will look at individual tournaments $T$, and the probability $\Pr(T)$ of having that tournament as the majority graph of a random election with a large number of voters $n \to \infty$. We also have some results on the relative weighting of the edges in a majority graph.

As remarked above, as the number $n$ of voters goes to infinity, the distribution of margin graphs approaches a multivariate normal distribution. Such a distribution can be analyzed using the covariance matrix which appears as a quadratic form in the probability density function. This matrix is computed in Theorem \ref{lem:cov-calc}. One can view the space of labeled directed graphs as a vector space acted on by the covariance matrix. The vector space of labeled directed graphs is known to split as a direct sum of the \textit{cycle space} and \textit{cut space}. We show in Theorem \ref{thm:eigen-calc} that these two spaces are exactly the eigenspaces of the covariance matrix of our distribution, and that moreover, the eigenvalue of the former is smaller than that of the latter. What this means is that:
\begin{center}
	\textit{The more cyclic a margin graph is, the less likely it is to arise.}
\end{center}
By this we mean that the probability density function takes lower values on margin graphs which are more cyclic.

Margin graphs are labeled, but to compute $\Pr(T)$ for a tournament $T$ we need to integrate the probability density function over all of the margin graphs compatible with $T$. Unfortunately, such probabilities are not well-understood for five or more candidates, but we formulate a conjecture which we check holds for up to five candidates.

\begin{definition}
	Given a tournament $T$ on $\ell$ vertices, we can assign a number to $T$ which we call the \textit{linearity} of $T$:
	\[ \lin(T) = \frac{1}{2} \sum_v \deg^-(v)^2 + \deg^+(v)^2 = \sum_v \deg^-(v)^2 = \sum_v \deg^+(v)^2.\]
\end{definition}

\noindent This value differs by a constant (depending on $\ell$) from the following value: for each pair of edges meeting at a common vertex $v$, add $1/2$ if the edge are either both into or both out of $v$, and subtract $1/2$ otherwise. Linearity is maximized by a linear order and is related the the decomposition into cycles and cuts. We conjecture that the more linear a tournament is, the more likely it is to arise as the majority graph of a random election.

{
\renewcommand{\thetheorem}{\ref{conj}}
\begin{conjecture}
	If $\lin(T) < \lin(T')$, then $\Pr(T) < \Pr(T')$.
\end{conjecture}
}

\noindent In Sections \ref{sec:eigenspaces}, \ref{sec:examples4}, and \ref{sec:examples} we compute the probability of obtaining each possible majority graph in the cases of three, four, and five candidates respectively, and find that the probabilities we compute agree with this conjecture.

\medskip{}

This paper was motivated by a question of Holliday and Pacuit posed in Holliday's 2019 seminar on Voting and Democracy at the University of California, Berkeley. Many voting methods are \textit{margin based} in the sense that the set of winners of an election depends only on the margin graph. Of these, some voting methods are \textit{majority based}; they chose a winner based solely on the majority graph, that is, depending only on which candidates defeat which other candidates head-to-head.

Holliday and Pacuit \cite{HollidayPacuitB} define an intermediate category of qualitative-margin based methods. We first need the following definition:
\begin{definition}
	A qualitative margin graph is a pair $(M,\prec)$ where $M$ is a majority graph and $\prec$ is a strict weak order on the set of edges of $M$. The qualitative margin graph of a profile $P$ is the pair $(\mc G(P);\prec_P)$ such
	that for any edges $(a,b)$ and $(c,d)$ in $\mc G(P)$, we have $(a,b) \prec_P (c,d)$ if $\Margin_P(a, b) < \Margin_P(c,d)$.
\end{definition}
\noindent A voting method is said to be \textit{qualitative-margin based} if it selects the set of winners based only on the qualitative margin graph. There are examples of voting methods which are qualitative-margin based but not majority based, such as the Simpson-Kramer Minimax method \cite{Simpson1969,Kramer1977}, Ranked Pairs \cite{Tideman1987}, weighted covering solutions \cite{Dutta1999,DeDonder2000}, Beat Path \cite{Schulze}, and Split Cycle \cite{HollidayPacuit,HollidayPacuitB}.

Holliday and Pacuit asked whether any qualitative margin graph, with the ordering on the edges being total, is obtained with positive probability even as the number of voters goes to infinity. We require that the ordering on edges be total because it is very unlikely, with large numbers of voters, to have two different margins be equal. We show that the answer to this question is positive:
{
	\renewcommand{\thetheorem}{\ref{thm:comparison2}}
\begin{theorem}
	Let $T$ be a tournament on a set of candidates $\mc{C}$ and let $\prec$ be an ordering of the edges of $T$. There is a number $N$ and a positive probability $p > 0$ such that: Given a set $\mc{V}$ of voters with $|\mc{V}| \geq N$, the probability is at least $p$ that the qualitative margin graph of a randomly chosen profile $P : \mc{V} \to \mc{L}(\mc{C})$ is $(T,\prec)$.
\end{theorem}}
\noindent Of course, this also means that any tournament is obtained as the majority graph with positive probability.
\begin{corollary}
	Let $T$ be a tournament on a set of candidates $\mc{C}$. There is a number $N$ and a positive probability $p > 0$ such that: Given a set $\mc{V}$ of voters with $|\mc{V}| \geq N$, the probability is at least $p$ that the majority graph of a randomly chosen profile $P : \mc{V} \to \mc{L}(\mc{C})$ is $T$.
\end{corollary}
\noindent As an immediate consequence of Theorem \ref{thm:comparison2}, most behaviours that a qualitative-margin based voting method could have with a small number of voters will also happen with positive probability even with large numbers of voters. (The only exception is that with small numbers of voters, there might be ties or two margins might be equal, and both of these situations are unlikely with large numbers of voters.) There are too many applications of this to list them all, but we will give two examples.

Holliday and Pacuit \cite{HollidayPacuit} were interested in whether their new voting method \textit{Split Cycle} satisfies a certain property of \textit{asymptotic resolvability} for $k > 3$ candidates. As a consequence of Theorem \ref{thm:comparison2}, it does not.
\begin{theorem}
	Split Cycle does not satisfy asymptotic resolvability for $k > 3$ candidates: As the number $n$ of voters approaches infinity, there is a positive probability that Split Cycle selects more than one winner.
\end{theorem}
\noindent See \cite{HollidayPacuit} for the proof of this using Theorem \ref{thm:comparison2}.

Another qualitative margin-based voting method is Minimax \cite{Simpson1969,Kramer1977}. This selects as the winners of an election all candidates whose greatest margin of defeat is the least among all candidates. It is possible for a candidate to be the Condorcet loser, defeated head-to-head by every other candidate, and yet still be chosen as the winner by Minimax. In the following margin graph $\mathsf{D}$ is the Condorcet loser, but its greatest loss is only by 5. Each other candidate is defeated by a margin of 7, 9, or 11, and so $\mathsf{D}$ is selected as the winner.
\[\xymatrix{&\mathsf{A}\ar[dl]_7\ar[dd]_(.65)3 \\\mathsf{B}\ar[rr]^(.35)9\ar[dr]_1&&\mathsf{C}\ar[ul]_{11}\ar[dl]_5\\
& \mathsf{D}}.\]
This margin graph corresponds to the qualitative margin graph
\[\xymatrix{&\mathsf{A}\ar[dl]_\gamma\ar[dd]_(.65)\sigma \\\mathsf{B}\ar[rr]^(.35)\beta\ar[dr]_\tau&&\mathsf{C}\ar[ul]_{\alpha}\ar[dl]_\rho\\
	& \mathsf{D}}\]
with $\alpha \succ \beta \succ \gamma \succ \rho \succ \sigma \succ \tau$.
We can then prove:
\begin{theorem}
	As the number of voters approaches infinity, Minimax has a positive probability of choosing as the unique winner the Condorcet loser.
\end{theorem}

\bigskip

The final section of this paper includes brief remarks on performing Monte Carlo simulations. Given a voting method, one might want to know, for example, how often it chooses a unique winner, or how often it deviates from some other voting method. Such calculations are performed for example by Holliday and Pacuit in \cite{HollidayPacuit}, where for $n$ voters they pick $n$ ballots at random and compute the margin graph. For large numbers of voters, one can instead pick a random margin graph from the multinomial distribution on margin graphs with large numbers of voters; this simply requires applying a linear transformation to independent normal random variables. We describe this method and give sample computations in Section \ref{sec:MonteCarlo}. This method is faster than a voter-by-voter simulation with a number of voters much larger than the number of candidates, and has already been used in \cite{HollidayPacuitC}.

\section{Probability background}\label{sec:prob}

We begin with a quick review of multivariate distributions and the central limit theorem.

\medskip

Let $\bfX = (X_1,\ldots,X_k)$ be a $k$-dimensional random vector. We will define what it means for $\bfX$ to have normal distribution $\bfX \sim \mathcal{N}_k(\bfmu,\bfSigma)$ with mean
\[ \bfmu = (\mu_1,\ldots,\mu_k) = \mathbb{E}[\bfX] \]
and $k \times k$ covariance matrix with entries
\[ \Sigma_{i,j} = \mathbb{E}[(X_i - \mu_i)(X_j - \mu_j)] = \Cov[X_i,X_j].\]
The covariance matrix is always positive semi-definite, and it is positive definite if and only if it is invertible.

If $\bfSigma$ is positive definite then we are in what is called the non-degenerate case, and the distribution has probability density function
\[ f_\bfX(\bfx) = \frac{e^{-\frac{1}{2}(\bfx - \bfmu)^T\bfSigma^{-1}(\bfx-\bfmu)}}{\sqrt{(2\pi)^k|\bfSigma|}}\]
where $|\bfSigma|$ is the determinant of $\bfSigma$. Note that the level sets of $f_\bfX$ are ellipses.

The book by Tong \cite{Tong} is a good reference for many properties of the multivariate normal distribution distribution.

\medskip

The importance of the multivariate normal distribution is due to the central limit theorem for multivariate random variables. Let $\bfX_1,\bfX_2,\bfX_3,\ldots$ be i.i.d.\ $k$-dimensional random vectors with mean $\bfmu$ and covariance matrix $\bfSigma$. Define
\[ \bfS_n = \frac{1}{n} \sum_{i=1}^{n} \bfX_i.\]
The multivariate central limit theorem says that $\sqrt{n}\left(\bfS_n - \bfmu \right)$ converges in distribution to the multivariate normal distribution with mean $0$ and covariance matrix $\bfSigma$:
\[ \sqrt{n}\left(\bfS_n - \bfmu\right) \overset{D}{\longrightarrow} \mc{N}_k(\mathbf{0},\bfSigma).\] What we mean by convergence in distribution is: For a sequence of random vectors $\bfX_1,\bfX_2,\ldots \in \mathbb{R}^k$, we say that this sequence converges in distribution to a random $k$-vector $\bfX$ if for each $A \subseteq \mathbb{R}^k$ which is a continuity set of $\bfX$ (i.e., $A$ is a Borel and $\Pr(\bfX \in \partial A) = 0$),
\[ \lim_{n \to \infty} \Pr(\bfX_n \in A) = \Pr(\bfX \in A).\]
See the book by van der Vaart \cite{vanderVaart} as a reference on the central limit theorem.

%\medskip
%
%The distributions that will be converging to The multinomial distribution is a generalization of the binomial distribution. Let $k$ be a fixed finite number. The multinomial distribution describes the outcome of performing $n$ trials, each with $k$ possible mutually exclusive outcomes, with probabilities $p_1,\ldots,p_k$ of each outcome. The probabilities $p_i$ satisfy $0 \leq p_i \leq 1$ and $\sum_{i=1}^k p_i = 1$. If $X_i$ is a random variable indicating the number of times the $i$th outcome is observed over $n$ trials, the $k$-vector $\bfX = (X_1,\ldots,X_k)$ follows the multinomial distribution with parameters $n$ and $P = (p_1,\ldots,p_n)$. The mean $\bfmu$ is $(np_1,\ldots,np_k)$. The covariance matrix of the multinomial distribution is a positive-semidefinite matrix $\boldsymbol{\Sigma}$ of rank $k-1$; this is because each random variable is a linear combination of the others: $X_1 + \cdots + X_k = n$. The covariance matrix $\bfSigma$ has entries
%\[ \Sigma_{i,i} = \Var(X_i) = np_i(1-p_i)\]
%and for $i \neq j$,
%\[ \Sigma_{i,j} = \Cov(X_i,X_j) = -n p_i p_j.\]
%If we drop the last variable $X_k$ and consider only the random $k-1$-vector $(X_1,\ldots,X_{k-1})$, then this has a positive definite covariance matrix.

\section{Cycle and cut spaces in graphs}\label{sec:spaces}

There is a homology theory for graphs in which one constructs for a graph $G$ and field $F$ a quite simple chain complex
\[ \cdots \to 0 \to \mc{E}_F(G) \overset{\partial}{\to} \mc{V}_F(G) \to 0 \to \cdots.\]
No knowledge of homology will be necessary as all of the objects we use are easily definable in purely graph-theoretic terms. This material is contained in most books covering algebraic graph theory, e.g., \cite{Biggs}.

Fix for the rest of this section a graph $G$, which in this paper will always be the complete graph $K_n$ on $n$ vertices, and the field $F = \mathbb{R}$. As is usual in homology, we must fix an orientation for each edge---for $K_n$ the complete graph on $v_1,\ldots,v_n$, we can think of an edge between $v_i$ and $v_j$ as a directed edge $v_i \to v_j$ for $i < j$.

The \textit{edge space} of $G = (V,E)$ is the $\mathbb{R}$-vector space $\mc E(G) = \mathbb{R}^E$. We can think of each element of $\mc E(G)$ as a formal sum $\sum_{e \in E} r_e e$, i.e., a labeling of each edge of the graph with a real number. The orientation of an edge $(u,v)$ is essentially just a choice to say that a positive number assigned to the edge is going from $u$ to $v$, and a negative number from $v$ to $u$; if we had chose the opposite orientation $(v,u)$, then a positive number assigned to the edge would be going from $v$ to $u$, and a negative number from $u$ to $v$. So if an edge has orientation $(u,v)$, we can define $(v,u) = -(u,v)$.

The \textit{vertex space} of $G = (V,E)$ is similarly defined using formal sums of vertices: it is the $\mathbb{R}$-vector space $\mc V(G) = \mathbb{R}^V$, and we think of each element of $\mc V(G)$ as a formal sum $\sum_{v \in V} r_v v$. There is also a boundary operator $\partial \colon \mc{E}(G) \to \mc{V}(G)$ which takes $(u,v)$ to $u - v$.

The homology group $H_1(G) = \ker \partial$ is easily seen to be the \textit{cycle space} of $G$. This is the subspace $\mc C(G)$ of $\mc E(G)$ consisting of all of those elements of the edge space with the property that, for each vertex, the sum of the numbers assigned to each incoming edge is equal to the sum of the numbers assigned to each outgoing edge; i.e., those elements $\sum_{e \in E} r_e e$ such that for every vertex $u$, 
\[ \sum_{(u,v) \in E} r_{(u,v)} - \sum_{(v,u) \in E} r_{(u,v)} = 0.\]
Using the convention $(v,u) = -(u,v)$ described above, we can just write
\[ \sum_{(u,v) \in E} r_{(u,v)} = 0.\]
If $u_1 \to u_2 \to \cdots \to u_n \to u_1$ is a cycle in $G$, then $(u_1,u_2) + (u_2,u_3) + \cdots + (u_n,u_1)$ is an element of the cycle space. In fact, the cycle space is generated by all such elements.

The \textit{cut space} of $G$ is the subspace $\mc C'(G)$ of $\mc E(G)$ generated by the edge cuts of $G$. A \textit{cut} of $G$ is a partition of the vertices of $G$ into two sets $S$ and $T$; the \textit{cut-set} of the cut is the set of edges with one end in $S$ and the other end in $T$. The cut space is generated by the elements
\[ \sum_{(u,v) \in E,\; u \in S,\; v \in T} (u,v) \]
where $(S,T)$ is a cut, and we again use the convention that $(v,u) = -(u,v)$.

Let $T$ be a spanning tree of $G$. We can use $T$ to compute bases for the cycle space and cut space. First, for each edge $e = (u,v)$ not in $T$, there is a unique path in $T$ from $u$ to $v$; together with $e$, this forms a cycle. We call the set of all such cycles the \textit{fundamental system of cycles} associated with $T$. It is not hard to see that these cycles are linearly independent, because each of them contains an edge not contained by any of the others. In fact, they form a basis for the cycle space. By counting, we see that
\[ \dim \mc C(G) = |E| -|V| + c(G) \]
where $c(G)$ is the number of connected components of $G$. For the cut set, it will be easiest to think about the case when $G$ is connected, and this will be the only case we use in the paper. Given an edge $e$ of the spanning tree $T$, $T - \{e\}$ splits into two connected components $S_1$ and $S_2$ which partition $V$, forming a cut of $G$ containing $e$ (and no other edge of $T$) in its edge set. The set of all such cuts forms a basis for the cycle space, and
\[ \dim \mc C'(G) = |V| - c(G).\]
Note that $\dim \mc C(G) + \dim \mc C'(G) = \dim \mc E(G)$. In fact, we can equip the edge space with the natural inner product, taking the edges as an orthonormal basis, and we have:
\begin{theorem}
	The cycle space and the cut space are orthogonal complements, so that $\mc E(G) = \mc C(G) \oplus \mc C'(G)$.
\end{theorem}

As mentioned above, in this paper we will always take $G$ to be the complete graph on $n$ vertices $v_1,\ldots,v_n$ and the orientation of each edge to be from $v_i$ to $v_j$ for $i < j$. As a spanning tree, we can take the tree
\[ \xymatrix{&&v_1\ar[dll]\ar[dl]\ar[d]\ar[drr]&&\\v_2&v_3&v_4&\cdots&v_n}\]
Then as a basis for the cycle space we can take the elements
\[ (v_1,v_i) + (v_i,v_j) + (v_j,v_1) = (v_1,v_i) + (v_i,v_j) - (v_1,v_j) \]
for $i < j$.
For the cut space, we take for each $i \geq 2$ the element
\[ \sum_{j \neq i} (v_i,v_j), \]
i.e., the sum of all outgoing edges from $v_i$.

\section{Random profiles}\label{sec:profiles}

We begin by fixing some notation for the rest of the paper. Earlier we fixed a set of $\ell$ candidates $\mc C = \{\mathsf{c}_1,\ldots,\mathsf{c}_\ell\}$. Let $G$ be the complete graph whose vertices are the candidates $\mc C$, and choose the natural orientation for the edges: $\mathsf{c}_i \to \mathsf{c}_j$ if $i < j$. We often write $i$ for $\mathsf{c}_i$ when we are thinking of it as a vertex of $G$ and write $\mathsf{c}_i$ only when we need to be especially clear.

Given a voter $\mathsf{v}$, under the IC assumption, $\mathsf{v}$ chooses a ballet, i.e., a ranking of the candidates, uniformly at random. We can associate with this a random vector $\mathbf{X}^{\mathsf{v}} = (X^{\mathsf{v}}_{i,j})_{i<j}$ with $X^{\mathsf{v}}_{i,j} = 1$ if $\mathsf{v}$ prefers $\mathsf{c}_i$ to $\mathsf{c}_j$, and $X_{i,j}^{\mathsf{v}}=-1$ if $\mathsf{v}$ prefers $\mathsf{c}_j$ to $\mathsf{c}_i$. We can think of $\mathbf{X}^{\mathsf{v}}$ as a random $\{1,-1\}$-valued element of the integral edge space $\mc{E}_{\mathbb{Z}}(G)$. For convenience, we write $X^{\mathsf{v}}_{j,i} = - X^{\mathsf{v}}_{i,j}$. Note that the mean is
$\bfmu = \mathbb{E}[\bfX^{\mathsf{v}}] = \mathbf{0}$. We next compute the covariance matrix $\bfSigma$ of $\bfX^{\mathsf{v}}$ by computing the covariances of the variables $X^{\mathsf{v}}_{i,j}$.

The covariance matrix $\bfSigma$ of $\bfX^{\mathsf{v}}$ is an $\frac{\ell (\ell - 1)}{2} \times \frac{\ell (\ell - 1)}{2}$ matrix whose rows and columns are indexed by edges of the graph $G$, with
\[ \Sigma_{(i,j);(r,s)} =  \Cov(X^{\mathsf{v}}_{i,j},X^{\mathsf{v}}_{r,s}).\]
The entries of the matrix $\bfSigma$ are indexed only by the edges of $G$ with the orientation we have fixed, i.e., with $i < j$ and $r < s$. We make the convention that we write $\Sigma_{(j,i);(r,s)} = - \Sigma_{(i,j);(r,s)}$ and $\Sigma_{(i,j);(s,r)} = - \Sigma_{(i,j);(r,s)}$. This is compatible with the convention for the $X^{\mathsf{v}}_{i,j}$ in the sense that we have
\[ \Sigma_{(i,j);(r,s)} =  \Cov(X^{\mathsf{v}}_{i,j},X^{\mathsf{v}}_{r,s}) \]
for arbitrary $i,j,r,s$.

\begin{lemma}\label{lem:cov-calc}
	We have, for $i,j,k,\ell$ all distinct:
	\begin{align*}
		\Sigma_{(i,j);(i,j)} &= \Cov(X^{\mathsf{v}}_{i,j},X^{\mathsf{v}}_{i,j}) = \Var(X^{\mathsf{v}}_{i,j}) = 1 \\
		\Sigma_{(i,j);(j,k)} &= \Cov(X^{\mathsf{v}}_{i,j},X^{\mathsf{v}}_{j,k}) = -\frac{1}{3}\\
		\Sigma_{(i,j);(k,\ell)} &= \Cov(X^{\mathsf{v}}_{i,j},X^{\mathsf{v}}_{k,\ell}) = 0.
	\end{align*}
\end{lemma}

\noindent Using our convention that $X^{\mathsf{v}}_{i,j} = - X^{\mathsf{v}}_{j,i}$, we see that
\[ \Sigma_{(i,j);(i,k)} = \Cov(X^{\mathsf{v}}_{i,j},X^{\mathsf{v}}_{i,k}) = \frac{1}{3} \qquad \text{ and } \qquad \Sigma_{(i,k);(j,k)} = \Cov(X^{\mathsf{v}}_{i,k},X^{\mathsf{v}}_{j,k}) = \frac{1}{3}.\]

\begin{proof}[Proof of Lemma \ref{lem:cov-calc}]
We omit the superscript $\mathsf{v}$ for this proof.

First,
\[ \Var(X_{i,j}) = \mathbb E(X_{i,j}^2) - \mathbb E(X_{i,j})^2 = 1 \]
since $X_{i,j}^2$ is always $1$, and in exactly half of the linear orders $L \in \mc{L}(C)$ we have $\mathsf{c}_i > _L\mathsf{c}_j$, so that $\mathbb E(X_{i,j}) = 0$.

For $i$, $j$, $r$, and $s$ all distinct, we have
\[ \Cov(X_{i,j},X_{r,s}) = \mathbb E\left[X_{i,j}\cdot X_{r,s}\right] - \mathbb{E}\left[X_{i,j}\right] \cdot \mathbb{E}\left[X_{r,s}\right] = \mathbb E\left[X_{i,j}\cdot X_{r,s}\right].\]
Then $X_{i,j}$ and $X_{r,s}$ are independent and have covariance $0$; indeed, among all of the linear orders $L \in \mc{L}(C)$ with $\mathsf{c}_i > _L\mathsf{c}_j$, exactly half of them have $\mathsf{c}_r > _L\mathsf{c}_s$, and the other half have $\mathsf{c}_s >_L \mathsf{c}_r$.

Finally, given $i,j,k$ distinct, we have
\[ \Cov(X_{i,j},X_{j,k}) = \mathbb E\left[X_{i,j}\cdot X_{j,k}\right] - \mathbb{E}\left[X_{i,j}\right] \cdot \mathbb{E}\left[X_{j,k}\right] = \mathbb E\left[X_{i,j}\cdot X_{j,k}\right].\]
Half of the linear orders have $\mathsf{c}_i > \mathsf{c}_j$. Of these, $\frac{1}{3}$ have $\mathsf{c}_i > \mathsf{c}_j > \mathsf{c}_k$, $\frac{1}{3}$ have $\mathsf{c}_i > \mathsf{c}_k > \mathsf{c}_j$, and $\frac{1}{3}$ have $\mathsf{c}_k > \mathsf{c}_i > \mathsf{c}_j$. Similarly, half the linear orders have $\mathsf{c}_j > \mathsf{c}_i$, and of these, $\frac{1}{3}$ have $\mathsf{c}_j > \mathsf{c}_i > \mathsf{c}_k$, $\frac{1}{3}$ have $\mathsf{c}_j > \mathsf{c}_k > \mathsf{c}_i$, and $\frac{1}{3}$ have $\mathsf{c}_k > \mathsf{c}_j > \mathsf{c}_i$. So
\[ \Cov(X_{i,j},X_{j,k}) = \mathbb E\left[X_{i,j}\cdot X_{j,k}\right] = \frac{1}{6} (1 -1 -1 -1 -1 +1) %= \frac{1}{6}\left(\frac{1}{4} - \frac{1}{4} - \frac{1}{4} - \frac{1}{4} - \frac{1}{4} + \frac{1}{4}  \right) 
= -\frac{1}{3}.\qedhere\]
\end{proof}

Now suppose we want to generate a random profile $P$ using $n$ voters $\mc{V} = \{\mathsf{v}_1,\ldots,\mathsf{v}_n\}$. If each voter chooses a ballot uniformly at random, then the random margin graph we obtain is a random variable equal to $\bfX^{\mathsf{v}_1} + \cdots + \bfX^{\mathsf{v}_n}$; this is again an element of the integral edge space $\mc{E}_{\mathbb{Z}}(G)$. The $\bfX^{\mathsf{v}_i}$ are i.i.d.\ with mean $\bfmu = \mathbf{0}$ and covariance matrix $\bfSigma$. As we see below in Theorem \ref{thm:inverse-calc} (where we compute the inverse) or in Theorem \ref{thm:eigen-calc} (where we compute the eigenvalues), $\bfSigma$ is invertible and hence positive semi-definite; we leave the proofs to later as they are heavily computational. Let $\bfS_n = \frac{1}{n} \sum_{i = 1}^n \bfX^{\mathsf{v}_i}$. By the central limit theorem, $\sqrt{n} \cdot \bfS_n$ converges in distribution to the multivariate normal distribution $\mc{N}(\mathbf{0},\bfSigma)$ with the same mean and covariance matrix, which has probability density function
\[ f(\mathbf{x}) = \frac{e^{-\frac{1}{2}\mathbf{x}^T\bfSigma^{-1}\mathbf{x}}}{\sqrt{(2\pi)^k|\bfSigma|}}.\]
That is, if $P_n$ is a random variable which represents a random profile given by $n$ voters choosing a ballot under the IC assumption, then $\frac{1}{\sqrt{n}} P_n$ converges in distribution to $\mc{N}(\mathbf{0},\bfSigma)$. Let $\bfY \sim \mc{N}(\mathbf{0},\bfSigma)$. Then we think of $\mathbf{Y}$ as an element of the real-valued edge space $\mc{E}_\mathbb{R}(G)$ which represents the margin graph of a random profile with a very large number of voters, with the edge weights normalized by a factor of $\sqrt{n}$. From this, without even making any further analysis of $\bfSigma$, we can already make several conclusions.
For example, the margin of victory of one candidate over another is usually on the order of $\sqrt{n}$.

The fact that $f(\bfx)$ is always positive for every $\bfx$ means that, given $\bfY\sim \mc{N}(\mathbf{0},\bfSigma)$ and any region $R \subseteq \mc{E}_{\mathbb{R}}(G)$ with positive measure, $\Pr(\bfY \in R) > 0$. As a consequence, for any tournament, we get a positive probability of that tournament being the majority graph of a random profile. Moreover, we can prove Theorem \ref{thm:comparison2} which says that the same is true for qualitative margin graphs.

\begin{theorem}\label{thm:comparison2}
	Let $T$ be a tournament on a set of candidates $\mc{C}$ and let $\prec$ be an ordering of the edges of $T$. There is a number $N$ and a positive probability $p > 0$ such that: Given a set $\mc{V}$ of voters with $|\mc{V}| \geq N$, the probability is at least $p$ that the qualitative margin graph of a randomly chosen profile $P : \mc{V} \to \mc{L}(\mc{C})$ is $(T,\prec)$.
\end{theorem}
\begin{proof}
	For a given set $\mc{V}$ of $n$ voters, randomly choose a profile $P : \mc{V} \to \mc{L}(\mc{C})$ by choosing, for each voter $\mathsf{v} \in \mc{V}$, a ballot $P(\mathsf{v}) \in \mc{L}(X)$ uniformly among all possible ballots with each ballot being chosen by a particular voter with probability $\frac{1}{n!}$. Taking the margin graph of $P$, and dividing by the number $n$ of voters, we generate a random vector $\mathbf{S}_n$ from $\mc{E}_{\mathbb{R}}(G)$ as above. As $n \to \infty$, $\sqrt{n} \cdot \mathbf{S}_n$ converges in distribution to the multinormal distrubtion with probability density function
	\[ f(\mathbf{x}) = \frac{\exp(-\frac{1}{2}(\mathbf{x} - \boldsymbol{\mu})^T\Sigma^{-1}(\mathbf{x}-\boldsymbol{\mu}))}{\sqrt{(2\pi)^k|\Sigma|}}.\]
	Let $\mathbf{Y} = (Y_{\mathsf{c}_i > \mathsf{c}_j})_{i < j}$ be a random variable with probability density function $f$. Since $\bfSigma$ is positive-definite, so is $\bfSigma^{-1}$; and so $f(\mathbf{x}) > 0$ for all $\mathbf{x}$.
	
	Write $\mathbf{S}_n = (S^n_{i,j})_{i < j}$. Now the probability that the qualitative margin graph of $P$ will be $T$ is the same as the probability that:
	\begin{itemize}
		\item for each edge $(\mathsf{c}_i,\mathsf{c}_j)$ of $T$, $S^n_{i,j} > 0$ (using the convention $S^n_{i,j} = - S^n_{j,i}$); and
		\item if the edge $(\mathsf{c}_i,\mathsf{c}_j)$ of $T$ is ranked higher by $\prec$ than the edge $(\mathsf{c}_s,\mathsf{c}_t)$, then $S^n_{i,j} > S^n_{s,t}$.
	\end{itemize}
	Note that these are all properties that are invariant under scaling; the margin graph of the random profile $P$ is actually $n \mathbf{S}_n$.
	
	As the number $n$ of voters goes to $\infty$, this probability converges to the probability that:
	\begin{itemize}
		\item for each edge $(\mathsf{c}_i,\mathsf{c}_j)$ of $T$, $Y_{i,j} > 0$;
		\item if the edge $(\mathsf{c}_i,\mathsf{c}_j)$ of $T$ is ranked higher by $\prec$ than the edge $(\mathsf{c}_s,\mathsf{c}_t)$, then $Y_{i,j} > Y_{s,t}$.
	\end{itemize}
	Because $\prec$ is a total order, this defines a region of the sample space of positive volume, and so this probability is positive as $f(\mathbf{x}) > 0$ for all $\mathbf{x}$.
\end{proof}

\bigskip{}

We return to computing the inverse of $\bfSigma$.

\begin{theorem}\label{thm:inverse-calc}
	The inverse of $\bfSigma$ is the matrix $\bfGamma = \bfSigma^{-1}$ with entries
	\[ \Gamma_{(i,j),(i,j)} =  \frac{3(\ell-1)}{\ell+1} \]
	\[ \Gamma_{(i,j),(i,k)} = -\frac{3}{\ell+1},\]
	and for $i,j,r,s$ all distinct,
	\[ \Gamma_{(i,j),(r,s)} = 0.\]
\end{theorem}

We again use the convention that $\Gamma_{(i,j),(r,s)} = - \Gamma_{(j,i),(r,s)} = - \Gamma_{(i,j),(s,r)} = \Gamma_{(j,i),(s,r)}$, so that, e.g.,
\[ \Gamma_{(i,j),(j,k)} = \frac{3}{\ell+1}.\]

\begin{proof}
	We check that $\Gamma \Sigma$ is the identity matrix by checking entry by entry. For $i < j$,
	\begin{align*}
	\left(\bfGamma \bfSigma \right)_{(i,j),(i,j)} &= \sum_{(r,s)} \Gamma_{(i,j),(r,s)} \cdot \Sigma_{(r,s),(i,j)} \\
	&= \Gamma_{(i,j),(i,j)} \cdot \Sigma_{(i,j),(i,j)}+  \sum_{k \neq i,j} \Gamma_{(i,j),(i,k)} \cdot \Sigma_{(i,k),(i,j)} + \sum_{k \neq i,j} \Gamma_{(i,j),(j,k)} \cdot \Sigma_{(j,k),(i,j)}\\
	&= 1 \cdot \frac{3(\ell-1)}{\ell+1} - 2 \cdot (\ell-2) \cdot \frac{1}{3} \cdot \frac{3}{\ell+1}\\
	&= 1.
	\end{align*} 
	Note that for $k < i$, the middle terms on the second line are really $\Gamma_{(i,j),(k,i)} \cdot \Sigma_{(k,i),(i,j)}$ but this is equal to $\Gamma_{(i,j),(i,k)} \cdot \Sigma_{(i,k),(i,j)}$, and similarly for the last terms when $k < j$. We will make the same kind of adjustments throughout the proof.
	
	For $i,j,k$ distinct, we have
	\begin{align*}
	\left(\bfGamma \bfSigma\right)_{(i,j),(i,k)} &= \sum_{(r,s)} \Gamma_{(i,j),(r,s)} \cdot \Sigma_{(r,s),(i,k)} \\
	&= \Gamma_{(i,j),(i,k)} \cdot \Sigma_{(i,k),(i,k)} + \Gamma_{(i,j),(i,j)} \cdot \Sigma_{(i,j),(i,k)} +  \Gamma_{(i,j),(j,k)} \cdot \Sigma_{(j,k),(i,k)} \\& \qquad +  \sum_{t \neq i,j,k} \Gamma_{(i,j),(i,t)} \cdot \Sigma_{(i,t),(i,k)} \\
	&= -1 \cdot \frac{3}{\ell+1} + \frac{1}{3} \cdot \frac{3(\ell-1)}{\ell+1} + \frac{1}{3} \cdot \frac{3}{\ell+1} - (\ell-3) \cdot \frac{3}{\ell+1} \cdot \frac{1}{3}
	\\&= 0
	\end{align*} 
	One of course must note, from the first line, that our normal convention still holds in the sense that, e.g., $\left(\bfGamma \bfSigma\right)_{(i,j),(i,k)} = -\left(\bfGamma \bfSigma\right)_{(j,i),(i,k)}$.
	
	Now given $i,j,r,s$ all distinct, we have:
	\begin{align*}
	\left(\bfGamma \bfSigma\right)_{(i,j),(r,s)} &= \sum_{(u,v)} \Gamma_{(i,j),(u,v)} \cdot \Sigma_{(u,v),(r,s)} \\
	&= \Gamma_{(i,j),(i,r)} \cdot \Sigma_{(i,r),(r,s)}
	+\Gamma_{(i,j),(i,s)} \cdot \Sigma_{(i,s),(r,s)}
	+ \Gamma_{(i,j),(j,r)} \cdot \Sigma_{(j,r),(r,s)}
	+\Gamma_{(i,j),)(j,s)} \cdot \Sigma_{(j,s),(r,s)}
	\\&= \frac{3}{\ell+1} \cdot \frac{1}{3} - \frac{3}{\ell+1} \cdot \frac{1}{3} - \frac{3}{\ell+1} \cdot \frac{1}{3} + \frac{3}{\ell+1} \cdot \frac{1}{3}
	\\&= 0
	\end{align*} 
	So $\bfGamma \bfSigma$ is the identity matrix, and $\bfGamma = \bfSigma^{-1}$.
\end{proof}

\section{Eigenvalues and eigenvectors}\label{sec:eigenspaces}

We know that each tournament has a positive probability of being obtained as the majority graph of a random election. But which tournaments are more likely than others?

Let $P_n$ be a random variable which represents a random margin graph given by $n$ voters choosing a ballot under the IC assumption. We know that $\frac{1}{\sqrt{n}} P_n$ converges in distribution to $\mc{N}(\mathbf{0},\bfSigma)$, with probability density function
\[ f(\mathbf{x}) = \frac{e^{-\frac{1}{2}\mathbf{x}^T\bfSigma^{-1}\mathbf{x}}}{\sqrt{(2\pi)^k|\bfSigma|}}.\]
So to understand what a random margin graph looks like with a large number of voters, we want to understand this distribution. We view this distribution as taking values in the real-valued edge space $\mc{E}_{\mathbb{R}}(G)$. In particular, we view $\bfSigma$ and its inverse $\bfSigma^{-1}$ as acting on $\mc{E}_{\mathbb{R}}(G)$.

The geometry of the probability density function is determined by the quadratic form $\mathbf{x}^T\bfSigma^{-1}\mathbf{x}$. We analyse this by computing its eigenvalues and eigenvectors. This is where the cycle space and cut space make an appearance.

Covariance matrices are always positive semidefinite, and as $\bfSigma$ is invertible it is positive definite. So the eigenvalues of $\bfSigma$ are all positive. The inverse $\bfSigma^{-1}$ has the same eigenspaces with reciprocal eigenvalues. Because $\bfSigma$ is symmetric, its eigenspaces will be orthogonal.

\begin{theorem}\label{thm:eigen-calc}
	$\bfSigma$ has two eigenvalues:
	\begin{itemize}
		\item $\lambda_1 = \frac{1}{12}$ with eigenspace the cycle space $\mc{C}_{\mathbb{R}}(G)$; and
		\item $\lambda_2 = \frac{1}{4} + \frac{\ell - 2}{12}$ with eigenspace the cut space $\mc{C}_{\mathbb{R}}'(G)$.
	\end{itemize}
\end{theorem}

Recall that the cycle space and the cut space are orthogonal. The inverse $\bfSigma^{-1}$ has the same eigenspaces with eigenvalues $1/\lambda_1$ and $1/\lambda_2$ respectively.

\begin{proof}[Proof of Theorem \ref{thm:eigen-calc}]
Consider a cycle $(\mathsf{c}_i,\mathsf{c}_j) + (\mathsf{c}_j,\mathsf{c}_k) + (\mathsf{c}_k,\mathsf{c}_i) \in \mc{E}_{\mathbb{R}}(G)$. Then $\bfSigma$ acts on this as
\[ \bfSigma \cdot \left( (\mathsf{c}_i,\mathsf{c}_j) + (\mathsf{c}_j,\mathsf{c}_k) + (\mathsf{c}_k,\mathsf{c}_i) \right) = \sum_{r < s} \left( \Sigma_{(i,j);(r,s)} + \Sigma_{(j,k);(r,s)} + \Sigma_{(k,i);(r,s)} \right) \cdot (\mathsf{c}_r,\mathsf{c}_s).\]
So the coefficient of any $(\mathsf{c}_r,\mathsf{c}_s)$ is
\[ \Sigma_{(i,j);(r,s)} + \Sigma_{(j,k);(r,s)} + \Sigma_{(k,i);(r,s)}.\]
The coefficient of $(\mathsf{c}_i,\mathsf{c}_j)$ is
\[ \Sigma_{(i,j);(i,j)} + \Sigma_{(j,k);(i,j)} + \Sigma_{(k,i);(i,j)} = 1 - \frac{1}{3} - \frac{1}{3} = \frac{1}{3}.\]
A similar argument works for $(\mathsf{c}_j,\mathsf{c}_k)$ and $(\mathsf{c}_k,\mathsf{c}_i)$. If $r$ and $s$ are distinct from $i,j,k$, then the coefficient of $(\mathsf{c}_r,\mathsf{c}_s)$ is clearly
\[ \Sigma_{(i,j);(r,s)} + \Sigma_{(j,k);(r,s)} + \Sigma_{(k,i);(r,s)} = 0+0+0 = 0.\]
If, say, $r = i$ and $s$ is distinct from $i,j,k$, then
\[ \Sigma_{(i,j);(i,s)} + \Sigma_{(j,k);(i,s)} + \Sigma_{(k,i);(i,s)} = \frac{1}{3} + 0 - \frac{1}{3} = 0.\]
So we conclude that
\[ \bfSigma \cdot \left( (\mathsf{c}_i,\mathsf{c}_j) + (\mathsf{c}_j,\mathsf{c}_k) + (\mathsf{c}_k,\mathsf{c}_i) \right) = \frac{1}{3} \left( (\mathsf{c}_i,\mathsf{c}_j) + (\mathsf{c}_j,\mathsf{c}_k) + (\mathsf{c}_k,\mathsf{c}_i) \right).\]
So one eigenvalue of $\bfSigma$ is $\frac{1}{3}$, and each element of the cycle space is an eigenvector.

\medskip{}

Now consider for a fixed $i$ the cut $\sum_{j \neq i} (\mathsf{c}_i,\mathsf{c}_j)$. $\bfSigma$ acts on this as
\[ \bfSigma \cdot \left( \sum_{j \neq i} (\mathsf{c}_i,\mathsf{c}_j) \right) = \sum_{r < s}\sum_{j \neq i}  \Sigma_{(i,j);(r,s)} (\mathsf{c}_r,\mathsf{c}_s). \]
For a fixed $k \neq i$, the coefficient of $(\mathsf{c}_i,\mathsf{c}_k)$ is
\[ \sum_{j \neq i}  \Sigma_{(i,j);(i,k)} = 1 +  \frac{\ell - 2}{3}.\]
Indeed, for $j = k$, $\Sigma_{(i,j);(i,k)} = 1$, and for $j \neq i,k$, $\Sigma_{(i,j);(i,k)} = \frac{1}{3}$.
For $r,s$ distinct from $i$, the coefficient of $(\mathsf{c}_r,\mathsf{c}_s)$ is
\[ \sum_{j \neq i}  \Sigma_{(i,j);(r,s)} = 0 \]
since, for $j = r$, $\Sigma_{(i,j);(r,s)} = \Sigma_{(i,r);(r,s)} = -\frac{1}{3}$, for $j = s$, $\Sigma_{(i,j);(r,s)} = \Sigma_{(i,s);(r,s)} = \frac{1}{3}$, and for $j \neq r,s$, $\Sigma_{(i,j);(r,s)} = 0$. So
\[ \bfSigma \cdot \left( \sum_{j \neq i} (\mathsf{c}_i,\mathsf{c}_j) \right) = \left( 1 + \frac{\ell-2}{3} \right) \cdot \left( \sum_{j \neq i} (\mathsf{c}_i,\mathsf{c}_j) \right).\]
Thus another eigenvalue of $\bfSigma$ is $1 + \frac{\ell-2}{3}$, and each element of the cut space is an eigenvector. Since the cut space and the cycle space are orthogonal complements, we have found all of the eigenvalues and eigenvectors.
\end{proof}

To get some visual picture of what is going on so far, let us consider the simplest case where there are only three candidates $\mathsf{A}$, $\mathsf{B}$, and $\mathsf{C}$. Each voter picks one of the six linear orders on these three candidates. Now $G$ is the graph on these three vertices, and we choose the natural orientation
\[ \xymatrix{& \mathsf{A}\ar[dl]\ar[dr] & \\
	\mathsf{B}\ar[rr] && \mathsf{C}}\]
The covariance matrix is
\[ \bfSigma = \renewcommand*{\arraystretch}{1.5}\begin{bmatrix} \phantom{-}1 & -\frac{1}{3} & \phantom{-}\frac{1}{3} \\
-\frac{1}{3} & \phantom{-}1 & \phantom{-}\frac{1}{3} \\
\phantom{-}\frac{1}{3} & \phantom{-}\frac{1}{3} & \phantom{-}1 \end{bmatrix}\]
with inverse
\[\bfSigma^{-1} = \renewcommand*{\arraystretch}{1.5}\begin{bmatrix}\phantom{-}\frac{3}{2} & \phantom{-}\frac{3}{4} & -\frac{3}{4} \\
\phantom{-}\frac{3}{4} & \phantom{-}\frac{3}{2} & -\frac{3}{4} \\
-\frac{3}{4} & -\frac{3}{4} & \phantom{-}\frac{3}{2} \end{bmatrix}.\]
The cycle space has dimension one and is generated by the cycle $(\mathsf{A},\mathsf{B}) + (\mathsf{B},\mathsf{C}) + (\mathsf{C},\mathsf{A})$. This is the eigenspace of $\bfSigma$ corresponding to the eigenvalue $\lambda_1 = \frac{1}{12}$, and the corresponding eigenvalue of $\bfSigma^{-1}$ is $12$. The cut space has dimension two and contains the elements $(\mathsf{A},\mathsf{B}) + (\mathsf{A},\mathsf{C})$, $(\mathsf{B},\mathsf{A}) + (\mathsf{B},\mathsf{C})$
, and $(\mathsf{C},\mathsf{A}) + (\mathsf{C},\mathsf{B})$. It is the eigenspace corresponding to the eigenvalue $\lambda_2 = \frac{1}{4} + \frac{\ell - 2}{12} = \frac{1}{3}$ of $\bfSigma$ and $3$ of $\bfSigma^{-1}$.

Now consider the probability density function
\[ f(\mathbf{x}) = \frac{e^{-\frac{1}{2}\mathbf{x}^T\bfSigma^{-1}\mathbf{x}}}{\sqrt{(2\pi)^k|\bfSigma|}}.\]
Write $\mathbf{x}$ as a linear combination of cycles and cuts, say $\mathbf{x} = \mathbf{y} + \mathbf{z}$ with $\mathbf{y} \in \mc{C}(G)$ an element of the cycle space and $\mathbf{z}\in \mc{C}'(G)$ an element of the cut space. The vectors $\mathbf{y}$ and $\mathbf{z}$ are orthogonal, and so
\[ \mathbf{x}^T\bfSigma^{-1}\mathbf{x} = \frac{1}{\lambda_1} ||\mathbf{y}||^2 + \frac{1}{\lambda_2} ||\mathsf{z}||^2 = 12 \cdot ||\mathbf{y}||^2 + 3 \cdot ||\mathbf{z}||^2.\]
The larger $\mathbf{x}^T\bfSigma^{-1}\mathbf{x}$ is, the smaller $f(\bfx)$ is; so, given that $||\bfx||^2 = ||\bfy||^2 + ||\bfz||^2$, for a fixed value of $||\bfx||$, $f(\bfx)$ is maximized for $\bfx$ in the cut space and minimized for $\bfx$ in the cycle space; and in general, the closer $\bfx$ is to being a cut, and the further it is from being a cycle, the larger $f(\bfx)$ is.

Now let us consider the geometry of the level sets of $f(\bfx)$. These level sets are ellipsoids, and the lengths of the axes are $1/\sqrt{\lambda_1} = 1/\sqrt{12}$ and $\sqrt{\lambda_2} = 1/\sqrt{3}$, with $\lambda_1 < \lambda_2$. So the cycle space corresponds to the minor axis of the ellipsoid, and the cut space corresponds to the major axis. We plot one of these level sets, $\bfx^T \bfSigma^{-1} \bfx = 1$, in the three-candidate case, in Figure \ref{fig:plot}. What we see is that the ellipsoid is shorter in two of the octants, as the minor axis of the ellipsoid is along the vector $(\mathsf{A},\mathsf{B}) + (\mathsf{B},\mathsf{C}) + (\mathsf{C},\mathsf{A})$ which is the diagonal vector of an octant.

\begin{figure}[h!]
	\includegraphics[width=0.49\textwidth]{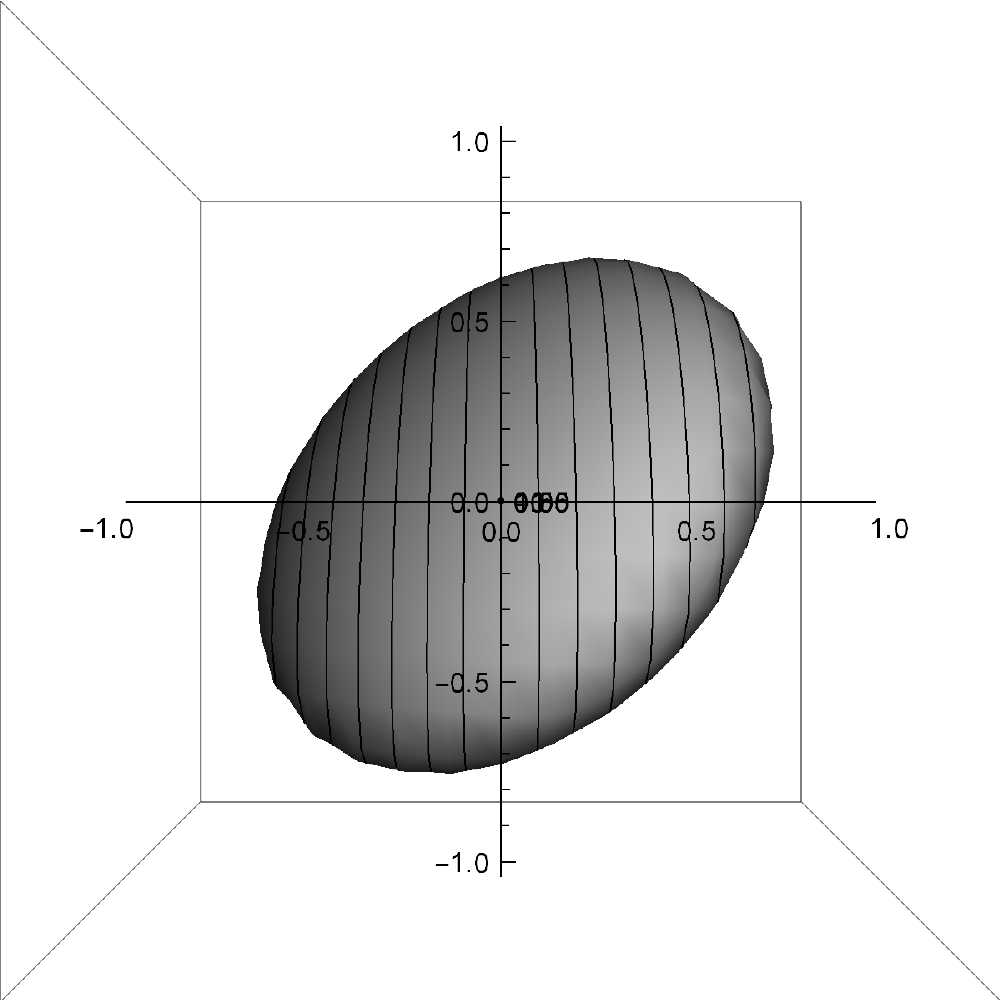}
	\includegraphics[width=0.49\textwidth]{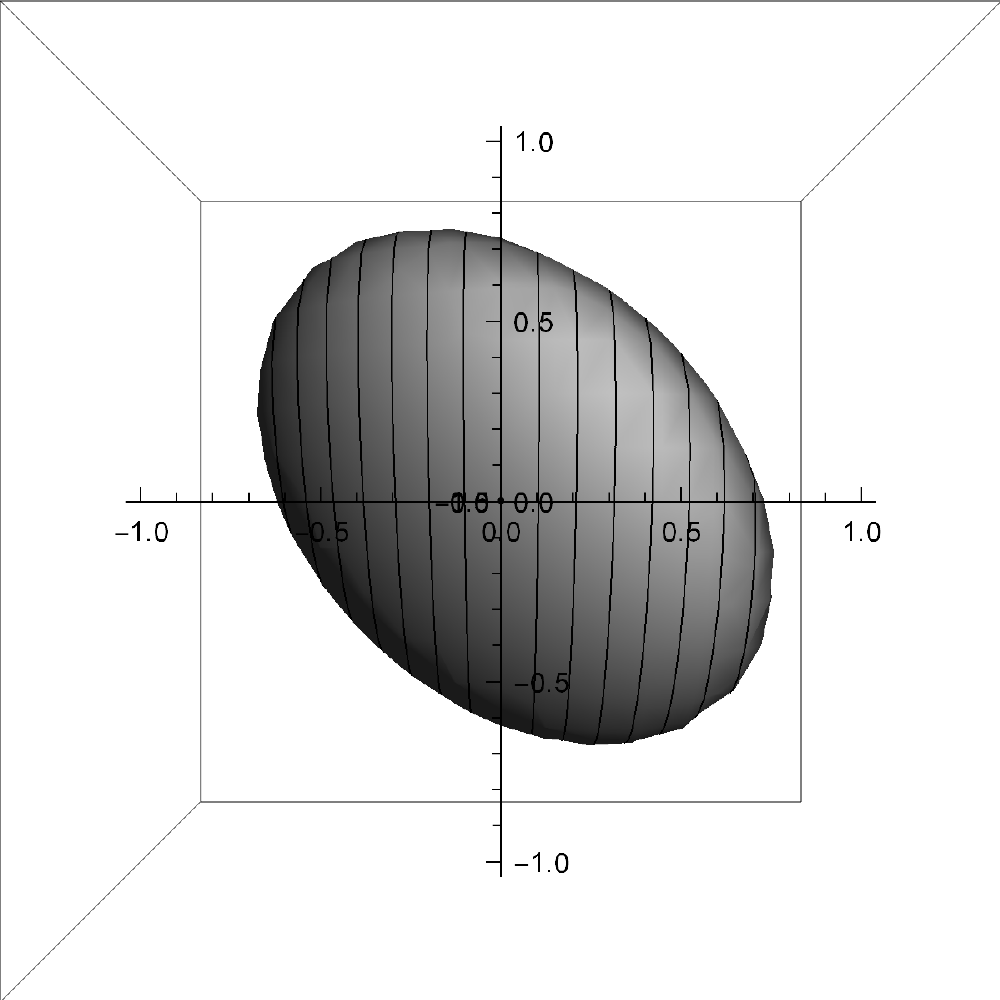}
	\includegraphics[width=0.49\textwidth]{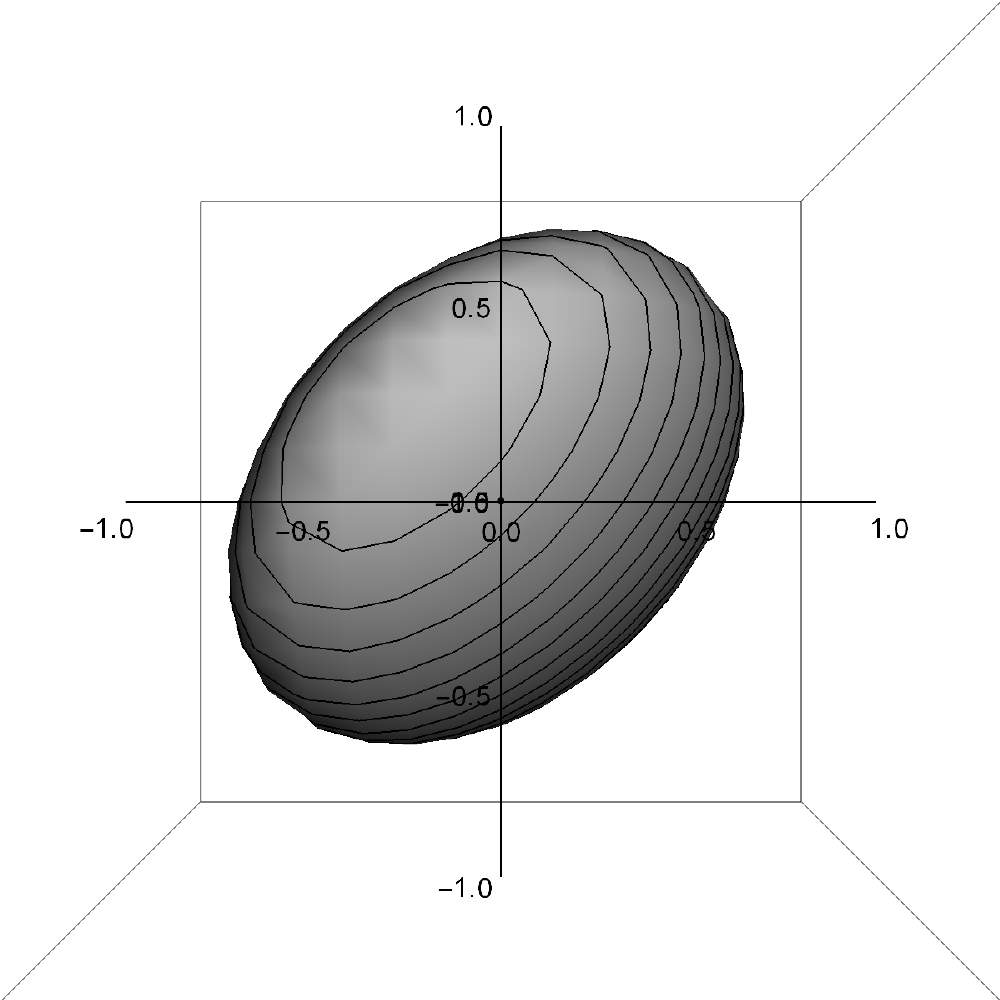}
	\includegraphics[width=0.49\textwidth]{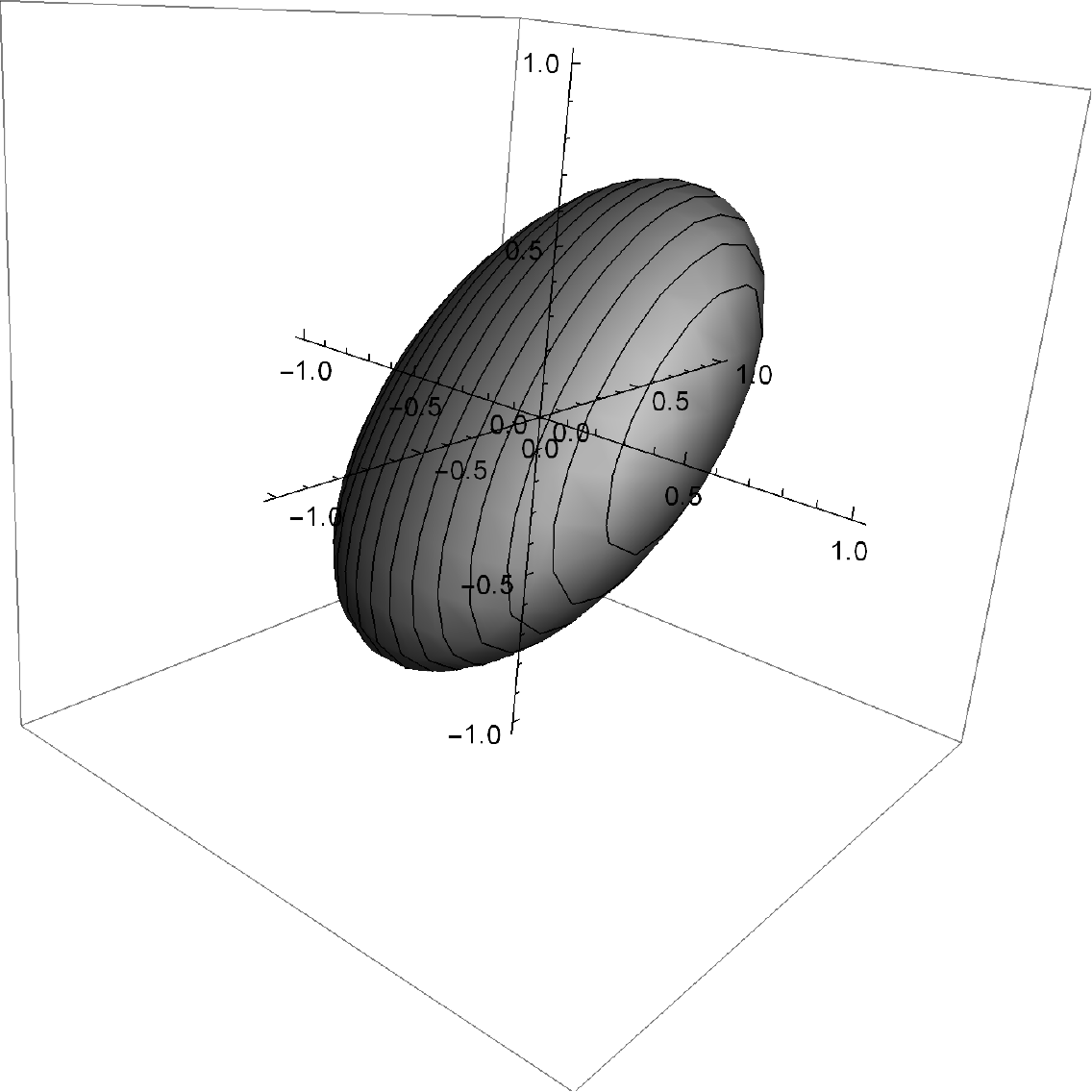}
	\caption{ $\bfx^T \bfSigma^{-1} \bfx = 1$.}
	\label{fig:plot}
\end{figure}

Now an octant corresponds to choosing a direction for each edge of the graph $G$, i.e., to a tournament on the candidates. There are eight tournaments on three candidates, the first six of which are the linear orders, and last two of which are cycles.

\[ \xymatrix{
	&& {}\save[]+<2.7cm,0cm>*\txt<8pc>{Linear orders} \restore &&&&&& {}\save[]+<3.6cm,0cm>*\txt<8pc>{Cycles} \restore \\
	& \mathsf{A} \ar[dl]\ar[dr] &      & & \mathsf{A} \ar[dr] &  & & \mathsf{A} \ar[dl] & & && \mathsf{A}\ar[dl] &\\
	\mathsf{B} \ar[rr] && \mathsf{C}         & \mathsf{B}\ar[ur]\ar[rr] & & \mathsf{C} & \mathsf{B} & & \mathsf{C} \ar[ll]\ar[ul] && \mathsf{B} \ar[rr] & & \mathsf{C} \ar[ul] \\
	& \mathsf{A} \ar[dl]\ar[dr] &      & & \mathsf{A} &  & & \mathsf{A} & & && \mathsf{A} \ar[dr] &\\
	\mathsf{B} && \mathsf{C} \ar[ll]         & \mathsf{B} \ar[ur] \ar[rr] & & \mathsf{C} \ar[ul] & \mathsf{B} \ar[ur] & & \mathsf{C} \ar[ll]\ar[ul] && \mathsf{B} \ar[ur] & & \mathsf{C} \ar[ll] \\ }\]
The minor axis of the ellipsoid is within the two octants corresponding to the cycles, and the major axis (which has dimension two) is orthogonal to this. Now this means that the volume of the ellipsoid within each of the two cycle octants is less than its volume within each of the other six. As $f(\mathbf{x})$ decreases as one gets further away, this means that the integral of $f(\mathbf{x})$ over each of the six octants will be larger than the integral of $f(\mathbf{x})$ over the two cycle octants.

This means that for a fixed tournament $T$, the probability that $T$ is the majority graph of a random profile is larger when $T$ is a linear order than it is when $T$ is a cycle. Since there are six linear orders and two cycles, the probability that the majority graph is a linear order (i.e., that there is a Condorcet winner) is at least $\frac{3}{4}$, and the probability that the majority graph is a cycle (i.e., that the paradox of voting occurs) is at most $\frac{1}{4}$.

In the case of three candidates, there are exact expressions for the probabilities of lying in a particular octant and hence of obtaining each of these tournaments \cite{NiemiWeisberg,GarmanKamien}. E.g., for
\[ \bfSigma = \begin{bmatrix}
1 & a & b \\
a & 1 & c \\
b & c & 1
\end{bmatrix}\]
the probability of being in the positive orthant is
\[ p = \frac{\arccos(-a)}{4\pi} + \frac{\arccos(-c)}{4\pi} - \frac{\arccos(b)}{4\pi}.\]
For each of the six linear orders, this gives us a probability of
\[ p = \frac{\arccos(-1/3)}{4\pi} = 0.1520\ldots \]
and for the each of the two cycles, the probability is
\[ p = \frac{\arccos(1/3)}{2\pi} - \frac{\arccos(-1/3)}{4\pi} = 0.0439 \ldots.\]
With more candidates, the same sort of analysis works, except that there are more possible tournaments, and the ellipsoid sits in a higher-dimensional space, e.g., with four candidates it sits in six dimensional space and with five candidates in ten dimensional space. Unfortunately there are exact expressions for orthant probabilities in higher dimensions only in very particular special cases.

\section{Decomposing into cycles plus cuts}\label{sec:decompose}

Recall that if we write $\mathbf{x} \in \mc{E}_{\mathbb{R}}(G)$ as a linear combination of cycles and cuts, say $\mathbf{x} = \mathbf{y} + \mathbf{z}$ with $\mathbf{y} \in \mc{C}(G)$ an element of the cycle space and $\mathbf{z}\in \mc{C}'(G)$ and element of the cut space, we have
\[ \mathbf{x}^T\bfSigma^{-1}\mathbf{x} = \frac{1}{\lambda_1} ||\mathbf{y}||^2 + \frac{1}{\lambda_2} ||\mathsf{z}||^2.\]
So it is natural to try to compute $||\bfy||^2$ and $||\bfz||^2$ from $\bfx$.

Given $\bfx = \sum_{(\mathsf{c}_i,\mathsf{c}_j)} x_{(\mathsf{c}_i,\mathsf{c}_j)} (\mathsf{c}_i,\mathsf{c}_j) \in \mc{E}_\mathbb{R}(G)$ and a candidate $\mathsf{c}_i$, we introduce the value
\[ \flow_\bfx(\mathsf{c}_i) = \sum_{(\mathsf{c}_i,\mathsf{c}_j)} x_{(\mathsf{c}_i,\mathsf{c}_j)}. \]
We think of this as the net flow out of the vertex $\mathsf{c}_i$, counting flow out of $\mathsf{c}_i$ positively and flow into $\mathsf{c}_i$ negatively. It is not hard to see that for a fixed $\mathsf{c}_i$, $\flow_\bfx(\mathsf{c}_i)$ is a linear function of $\bfx$. Moreover, if $\bfx$ is part of the cycle space, then $\flow_\bfx(\mathsf{c}_i) = 0$; one can easily see this by checking that it is true for the basic cycles $(\mathsf{c}_i,\mathsf{c}_j) + (\mathsf{c}_j,\mathsf{c}_k) + (\mathsf{c}_k,\mathsf{c}_i)$ and extends linearly to the whole cycle space.

Then we can compute:

\begin{theorem}\label{thm:calc-cut-part}
	Let $\bfx \in \mc{E}_{\mathbb{R}}(G)$ be $\bfx = \sum_{(\mathsf{c}_i,\mathsf{c}_j)} x_{(\mathsf{c}_i,\mathsf{c}_j)} (\mathsf{c}_i,\mathsf{c}_j)$. Write $\mathbf{x} = \mathbf{y} + \mathbf{z}$ with $\mathbf{y} \in \mc{C}(G)$ an element of the cycle space and $\mathbf{z}\in \mc{C}'(G)$ an element of the cut space. Then
	\[ \bfz = \sum_{(\mathsf{c}_i,\mathsf{c}_j)} \frac{\flow_\bfx(\mathsf{c}_i) - \flow_\bfx(\mathsf{c}_j)}{\ell} (\mathsf{c}_i,\mathsf{c}_j).\]
	Moreover,
	\[ ||\bfz||^2 = \frac{2}{\ell} \cdot \left(\sum_{(\mathsf{c}_i,\mathsf{c}_j)} x_{(\mathsf{c}_i,\mathsf{c}_j)}^2 + \sum_{(\mathsf{c}_i,\mathsf{c}_j),(\mathsf{c}_i,\mathsf{c}_k)} x_{(\mathsf{c}_i,\mathsf{c}_j)} x_{(\mathsf{c}_i,\mathsf{c}_k)} \right).\]
	where the convention $x_{(\mathsf{c}_i,\mathsf{c}_j)} = - x_{(\mathsf{c}_j,\mathsf{c}_i)}$ handles the sign of the terms $x_{(\mathsf{c}_i,\mathsf{c}_j)} x_{(\mathsf{c}_i,\mathsf{c}_k)}$.
\end{theorem}
\begin{proof}
Since flow is a linear operator, and $\bfy$ is in the cycle space, we have for each vertex $\mathsf{c}_i$ that
\[ \flow_\bfx(\mathsf{c}_i) = \flow_\bfy(\mathsf{c}_i) + \flow_\bfz(\mathsf{c}_i) = \flow_\bfz(\mathsf{c}_i).\]
Now writing $\bfz = \sum_{(\mathsf{c}_i,\mathsf{c}_j)} z_{(\mathsf{c}_i,\mathsf{c}_j)} (\mathsf{c}_i,\mathsf{c}_j)$ we argue that
\[ \flow_\bfz(\mathsf{c}_i) - \flow_\bfz(\mathsf{c}_j) = \ell z_{(\mathsf{c}_i,\mathsf{c}_j)}.\]
It suffices by linearity to show that this is true for a generating set of the cut space, i.e., that it is true for each of the cuts $\sum_{\mathsf{c}_k} (\mathsf{c}_i,\mathsf{c}_k)$, $\sum_{\mathsf{c}_k} (\mathsf{c}_i,\mathsf{c}_k)$, and for a fixed $k* \neq i,j$, $\sum_{\mathsf{c}_k} (\mathsf{c}_{k^*},\mathsf{c}_k)$. For $\bfz = \sum_{\mathsf{c}_k} (\mathsf{c}_i,\mathsf{c}_k)$, we have
\[ \flow_\bfz(\mathsf{c}_i) - \flow_\bfz(\mathsf{c}_j) = (\ell-1) - (-1) = \ell,\]
and the coefficient of $(\mathsf{c}_i,\mathsf{c}_j)$ in $\bfz$ is $1$. A similar argument works for $\bfz = \sum_{\mathsf{c}_k} (\mathsf{c}_j,\mathsf{c}_k)$. For $\bfz = \sum_{\mathsf{c}_k} (\mathsf{c}_{k^*},\mathsf{c}_k)$, we have
\[ \flow_\bfz(\mathsf{c}_i) - \flow_\bfz(\mathsf{c}_j) = -1 - (-1) = 0, \]
and the coefficient of $(\mathsf{c}_i,\mathsf{c}_j)$ in $\bfz$ is $0$. Thus we have
\[ \flow_\bfx(\mathsf{c}_i) - \flow_\bfx(\mathsf{c}_j) = \ell z_{(\mathsf{c}_i,\mathsf{c}_j)}\]
and so
\[ \bfz = \sum_{(\mathsf{c}_i,\mathsf{c}_j)} \frac{\flow_\bfx(\mathsf{c}_i) - \flow_\bfx(\mathsf{c}_j)}{\ell} (\mathsf{c}_i,\mathsf{c}_j).\]
This completes the first part of the theorem. 

%For the second part of the theorem, we could use a direct computation using the first part, but it follows more directly from the computation of $\bfSigma^{-1}$ in Theorem \ref{thm:inverse-calc} and of the eigenspaces of $\bfSigma^{-1}$ in Theorem \ref{thm:eigen-calc}. Indeed, from Theorem \ref{thm:eigen-calc} we get that
%\[ \bfx \bfSigma^{-1} \bfx = 12 ||y||^2 + \frac{12}{3 + \frac{\ell-2}} ||z||^2 \]

Now we must compute $||\bfz||^2$. We have
\begin{align*}
||\bfz||^2 &= \sum_{\mathsf{c}_i,\mathsf{c}_j} z_{(\mathsf{c}_i,\mathsf{c}_j)}^2 \\
&= \sum_{(\mathsf{c}_i,\mathsf{c}_j)} \frac{(\flow_\bfx(\mathsf{c}_i) - \flow_\bfx(\mathsf{c}_j))^2}{\ell^2} \\
&= \frac{1}{\ell^2} \sum_{(\mathsf{c}_i,\mathsf{c}_j)} \left(\sum_{\mathsf{c}_r} x_{(\mathsf{c}_i,\mathsf{c}_r)} - \sum_{\mathsf{c}_s} x_{(\mathsf{c}_j,\mathsf{c}_s)}\right)^2.
\end{align*}
After expanding this out, we will get a linear combination of terms of the form $x_{(\mathsf{c}_i,\mathsf{c}_{r})} \cdot x_{(\mathsf{c}_j,\mathsf{c}_{s})}$.

First consider terms of the form $x_{(\mathsf{c}_t,\mathsf{c}_u)} \cdot x_{(\mathsf{c}_v,\mathsf{c}_w)}$ with $t,u,v,w$ all distinct. Such a term shows up four times after expanding, for $(\mathsf{c}_i,\mathsf{c}_j)$ equal to each of $(\mathsf{c}_t,\mathsf{c}_v)$, $(\mathsf{c}_t,\mathsf{c}_w)$, $(\mathsf{c}_u,\mathsf{c}_v)$, and $(\mathsf{c}_u,\mathsf{c}_w)$. It shows up with the opposite sign for $t$ as compared to $u$, and $v$ as compared to $w$. So these all cancel out and the coefficient is zero.

Next, consider terms of the form $x_{(\mathsf{c}_u,\mathsf{c}_v)} \cdot x_{(\mathsf{c}_u,\mathsf{c}_w)}$. Such a term appears for:
\begin{itemize}
	\item $(\mathsf{c}_i,\mathsf{c}_j) = (\mathsf{c}_u,\mathsf{c}_v)$: For $\mathsf{c}_r = \mathsf{c}_v$ we get $x_{(\mathsf{c}_i,\mathsf{c}_r)} = x_{(\mathsf{c}_u,\mathsf{c}_v)}$, and for $\mathsf{c}_r = \mathsf{c}_w$ we get $ x_{(\mathsf{c}_i,\mathsf{c}_r)} = x_{(\mathsf{c}_u,\mathsf{c}_w)}$; and for $\mathsf{c}_s = \mathsf{c}_u$ we get $ x_{(\mathsf{c}_j,\mathsf{c}_s)} = x_{(\mathsf{c}_v,\mathsf{c}_u)} = -x_{(\mathsf{c}_u,\mathsf{c}_v)}$. After squaring we get a term $4 \cdot x_{(\mathsf{c}_u,\mathsf{c}_v)} \cdot x_{(\mathsf{c}_u,\mathsf{c}_w)}$.
	\item $(\mathsf{c}_i,\mathsf{c}_j) = (\mathsf{c}_u,\mathsf{c}_w)$: Similarly to the above, after squaring we get a term $4 \cdot x_{(\mathsf{c}_u,\mathsf{c}_v)} \cdot x_{(\mathsf{c}_u,\mathsf{c}_w)}$.
	\item $\mathsf{c}_i = \mathsf{c}_u$, $\mathsf{c}_j \neq \mathsf{c}_v,\mathsf{c}_w$: For $\mathsf{c}_r = \mathsf{c}_v$ we get $x_{(\mathsf{c}_i,\mathsf{c}_r)} = x_{(\mathsf{c}_u,\mathsf{c}_v)}$, and for $\mathsf{c}_r = \mathsf{c}_w$ we get $ x_{(\mathsf{c}_i,\mathsf{c}_r)} = x_{(\mathsf{c}_u,\mathsf{c}_w)}$; after squaring, we get $2 \cdot x_{(\mathsf{c}_u,\mathsf{c}_v)} \cdot x_{(\mathsf{c}_u,\mathsf{c}_w)}$.
	\item $(\mathsf{c}_i,\mathsf{c}_j) = (\mathsf{c}_v,\mathsf{c}_w)$: For $\mathsf{c}_r = \mathsf{c}_u$ we get $x_{(\mathsf{c}_i,\mathsf{c}_r)} = x_{(\mathsf{c}_v,\mathsf{c}_u)} = - x_{(\mathsf{c}_v,\mathsf{c}_u)}$, and for $\mathsf{c}_s = \mathsf{c}_u$ we get $ x_{(\mathsf{c}_j,\mathsf{c}_s)} = x_{(\mathsf{c}_w,\mathsf{c}_u)} = - x_{(\mathsf{c}_w,\mathsf{c}_u)}$; after squaring, we get $- 2 \cdot x_{(\mathsf{c}_u,\mathsf{c}_v)} \cdot x_{(\mathsf{c}_u,\mathsf{c}_w)}$.
\end{itemize}
There are $\ell-3$ instances of the third case and one instance of each other case. So the coefficient of $x_{(\mathsf{c}_u,\mathsf{c}_v)} \cdot x_{(\mathsf{c}_u,\mathsf{c}_w)}$ is $4 + 4 + 2\ell-6 - 2 = 2\ell$. There are four other similar cases, such as $x_{(\mathsf{c}_u,\mathsf{c}_v)} \cdot x_{(\mathsf{c}_v,\mathsf{c}_w)}$. In this case for example, the coefficient will be $-2\ell$. The sign is positive if the edges are either both into or both out of the same common vertex, and negative otherwise.

Finally, consider terms of the form $x_{(\mathsf{c}_u,\mathsf{c}_v)}^2$. This term appears, with a coefficient of $4$, for $(\mathsf{c}_i,\mathsf{c}_j) = (\mathsf{c}_u,\mathsf{c}_v)$, and with a coefficient of $1$ for each $(\mathsf{c}_i,\mathsf{c}_j)$ which has exactly one vertex in common with $(\mathsf{c}_u,\mathsf{c}_v)$, of which there are $2(\ell-2)$. So $x_{(\mathsf{c}_u,\mathsf{c}_v)}^2$ has a coefficient of $2\ell$.

So:
\[ ||\bfz||^2 = \frac{2}{\ell} \cdot \left(\sum_{(\mathsf{c}_i,\mathsf{c}_j)} x_{(\mathsf{c}_i,\mathsf{c}_j)}^2 + \sum_{(\mathsf{c}_i,\mathsf{c}_j),(\mathsf{c}_i,\mathsf{c}_k)} x_{(\mathsf{c}_i,\mathsf{c}_j)} x_{(\mathsf{c}_i,\mathsf{c}_k)} \right). \qedhere\]
\end{proof}

Now let us apply this to our scenario. We have
\[ f(\mathbf{x}) = \frac{e^{-\frac{1}{2}\mathbf{x}^T\bfSigma^{-1}\mathbf{x}}}{\sqrt{(2\pi)^k|\bfSigma|}},\]
where
\[ \mathbf{x}^T\bfSigma^{-1}\mathbf{x} = \frac{1}{\lambda_1} ||\mathbf{y}||^2 + \frac{1}{\lambda_2} ||\mathsf{z}||^2. \]
Then we can write
\[ f(\mathbf{x}) = \frac{e^{-\frac{1}{2 \lambda_1}||\mathbf{y}||^2 - \frac{1}{2\lambda_2} ||\mathsf{z}||^2}}{\sqrt{(2\pi)^k|\bfSigma|}} = \frac{e^{-\frac{1}{2 \lambda_1}||\mathbf{x}||^2 + \left(\frac{1}{2\lambda_1} - \frac{1}{2\lambda_2}\right) ||\mathsf{z}||^2}}{\sqrt{(2\pi)^k|\bfSigma|}}.\]
Since $\frac{1}{2\lambda_1} - \frac{1}{2\lambda_2} \geq 0$, this makes it clear that for a specific value of $||\bfx||$, the larger $||\bfz||$ is, the larger $f(\bfx)$ is.
Using Theorem \ref{thm:calc-cut-part}, we have
\[ ||\bfz||^2 = \frac{2}{\ell} \cdot \left(\sum_{(\mathsf{c}_i,\mathsf{c}_j)} x_{(\mathsf{c}_i,\mathsf{c}_j)}^2 + \sum_{(\mathsf{c}_i,\mathsf{c}_j),(\mathsf{c}_i,\mathsf{c}_k)} x_{(\mathsf{c}_i,\mathsf{c}_j)} x_{(\mathsf{c}_i,\mathsf{c}_k)} \right).\]
Given a tournament $T$ on the candidates $\mathcal{C}$, the probability that we obtain $T$ as the majority graph of a random profile $\bfX \sim \mc{N}(\mathbf{0},\bfSigma)$ is
\[ \Pr(T) = \int_{R_T} f(\bfx) \; d\bfx \]
where $R_T$ is the orthant corresponding to $T$; $R_T$ is the octant which has $x_{(\mathsf{c}_i,\mathsf{c}_j)}$ positive if and only if $\mathsf{c}_i \to \mathsf{c}_j$ in $T$. (Note that by our convention that $x_{(\mathsf{c}_i,\mathsf{c}_j)} = -x_{(\mathsf{c}_j,\mathsf{c}_i)}$, if $i < j$ and $j \to i$ in $T$, then the region $R_T$ has $x_{(\mathsf{c}_i,\mathsf{c}_j)}$ negative which means that $x_{(\mathsf{c}_j,\mathsf{c}_i)}$ is positive.) So we can write
\begin{align*}
||\bfz||^2 = \frac{2}{\ell} \cdot & \Big(\sum_{\mathsf{c}_i \to \mathsf{c}_j} x_{(\mathsf{c}_i,\mathsf{c}_j)}^2 + \sum_{\mathsf{c}_i \to \mathsf{c}_j, \; \mathsf{c}_i \to \mathsf{c}_k} x_{(\mathsf{c}_i,\mathsf{c}_j)} x_{(\mathsf{c}_i,\mathsf{c}_k)} \\ &+ \sum_{\mathsf{c}_j \to \mathsf{c}_i, \; \mathsf{c}_k \to \mathsf{c}_i} x_{(\mathsf{c}_j,\mathsf{c}_i)} x_{(\mathsf{c}_k,\mathsf{c}_i)} - \sum_{\mathsf{c}_i \to \mathsf{c}_j, \; \mathsf{c}_k \to \mathsf{c}_i} x_{(\mathsf{c}_i,\mathsf{c}_j)} x_{(\mathsf{c}_k,\mathsf{c}_i)} \Big)
\end{align*} 
where $\to$ is interpreted as in $T$. In the region $R_T$, all of the terms $x_{(\cdot,\cdot)}$ appearing in this expression will be positive. Intuitively, on might expect that the more positive terms show up, the greater the integral 
\[ \Pr(T) = \int_{R_T} f(\bfx) \; d\bfx \]
should be. Moreover, our geometric intuition agrees. Suppose that we have an ellipsoid such that all of the axes take on one of two values (dividing the axes into major axes and minor axes). Suppose also that we have an orthant. It seems plausible that the closer the diagonal vector in the orthant is to the major axes of the ellipsoid, the greater the volume within that orthant is. (Unfortunately, volumes of octants of ellipsoids are directly related to the measurement of solid angles of simplicial cones, and this is not understood in higher dimensions. See, e.g., \cite{Ribando}.)

To this effect, we assign a number to each tournament:
\begin{definition}
	Given a tournament $T$ on $\ell$ vertices, we can assign a number to $T$ which we call the \textit{linearity} of $T$:
	\[ \lin(T) = \frac{1}{2} \sum_v \deg^-(v)^2 + \deg^+(v)^2 = \sum_v \deg^-(v)^2 = \sum_v \deg^+(v)^2. \]
\end{definition}

\noindent Then the number of positive terms in $||\bfz||^2$ on the octant $R_T$ is exactly $\lin(T) - \ell(\ell-1)$ as
\[ \sum_{v} \deg^-(v) \cdot (\deg^-(v) - 1) + \deg^+(v) \cdot (\deg^+(v) - 1) = \lin(T) - \ell(\ell-1).\]
So we have a guess at what is going on: The greater the linearity of $T$, the greater the probability $\Pr(T)$ should be.

\begin{conjecture}\label{conj}
	Let $T$ and $T'$ be tournaments on $\ell$ candidates. If $\lin(T) < \lin(T')$, then $\Pr(T) < \Pr(T')$.
\end{conjecture}

\noindent This conjecture has already been verified for $\ell = 3$ (as the linearity of a linear order is higher than the linearity of a cycle on three candidates). In the next two sections we verify the conjecture for $\ell = 4$ and $\ell = 5$ as well.

\section{Four Candidates}\label{sec:examples4}

Consider the case of four candidates $\mathsf{A}$, $\mathsf{B}$, $\mathsf{C}$, and $\mathsf{D}$. As usual we fix by convention the edge directions of the linear order $\mathsf{A} > \mathsf{B}> \mathsf{C}> \mathsf{D}$.
\[ \xymatrix{
	\mathsf{A}\ar[r]\ar[d]\ar[dr] & \mathsf{B} \ar[dl]\ar[d] \\
	\mathsf{C} \ar[r] & \mathsf{D}
}\]
Then we compute
\[ \bfSigma = \begin{blockarray}{ccccccc}
& \mathsf{AB} & \mathsf{AC} & \mathsf{AD} & \mathsf{BC} & \mathsf{BD} & \mathsf{CD} \\
\begin{block}{c[cccccc]}
\mathsf{AB} & \phantom{-}1 & \phantom{-}1/3 & \phantom{-}1/3 & -1/3 & -1/3 & \phantom{-}0 \\
\mathsf{AC} & \phantom{-}1/3 & \phantom{-}1 & \phantom{-}1/3 & \phantom{-}1/3 & \phantom{-}0 & -1/3 \\
\mathsf{AD} & \phantom{-}1/3 & \phantom{-}1/3 & \phantom{-}1 & \phantom{-}0 & \phantom{-}1/3 & \phantom{-}1/3 \\
\mathsf{BC} & -1/3 & \phantom{-}1/3 & \phantom{-}0 & \phantom{-}1 & \phantom{-}1/3 &- 1/3 \\
\mathsf{BD} & -1/3 & \phantom{-}0 & \phantom{-}1/3 & \phantom{-}1/3 & \phantom{-}1 & \phantom{-}1/3 \\
\mathsf{CD} & \phantom{-}0 & -1/3 & \phantom{-}1/3 & -1/3 & \phantom{-}1/3 & \phantom{-}1 \\
\end{block}
\end{blockarray}.\]
For four candidates, there are four possible isomorphism types:
\[ \xymatrix{
 {}\save[]+<0.5cm,-0.6cm>*\txt<8pc>{$T_1$:} \restore & \mathsf{A}\ar[r]\ar[d]\ar[dr] & \mathsf{B} \ar[dl]\ar[d] & {}\save[]+<0.5cm,-0.6cm>*\txt<8pc>{$T_2$:}\restore & \mathsf{A}\ar[r]\ar[d]\ar[dr] & \mathsf{B} \ar[dl] & {}\save[]+<0.5cm,-0.6cm>*\txt<8pc>{$T_3$:}\restore & \mathsf{A}\ar[r]\ar[dr] & \mathsf{B} \ar[d]\ar[dl] & {}\save[]+<0.5cm,-0.6cm>*\txt<8pc>{$T_4$:}\restore & \mathsf{A} \ar[r]\ar[dr] & \mathsf{B} \ar[d]\ar[dl] \\
&\mathsf{C} \ar[r] & \mathsf{D} && \mathsf{C} \ar[r] & \mathsf{D}\ar[u] && \mathsf{C} \ar[u]\ar[r] & \mathsf{D} && \mathsf{C}\ar[u] & \mathsf{D} \ar[l]
}\]
There are twenty-four different labeled copies of the first one, eight of the second and third, and twenty-four of the last.

We have a new phenomenon that shows up here, which is that $T_2$ and $T_3$ are \textit{dual} in the sense that one is obtained from the other by reversing the directions of every arrow. Thus the probability of obtaining $T_2$ will be the same as the probability of obtaining $T_3$. $T_1$ and $T_4$ are self-dual in the sense that by reversing arrows, we keep the same isomorphism type. Dual tournaments have the same linearity.

To compute the probabilities of obtaining each of these tournaments directly, one would have to find orthant probabilities in six dimensions, and there is no known general way of doing this. However, a clever argument by Gehrlein and Fishburn \cite{GehrleinFishburn78} allows us to compute exact probabilities. 

First, the probability of having $\mathsf{A}$ as a majority winner is the probability that $\mathsf{A}$ beats $\mathsf{B}$, $\mathsf{C}$, and $\mathsf{D}$; thus we have reduced ourselves to a three-dimensional space and can compute the probability exactly. By multiplying by four, we get the probability of having a majority winner, which is the same as the probability that the margin graph of a random tournament is isomorphic to either $T_1$ or $T_2$. (This value is 0.82452.)

Second (and this is the clever part) the probability that the margin graph of a random election is transitive (i.e., isomorphic to $T_1$) can be computed as follows. First, consider the probability that $\mathsf{A}$ beats $\mathsf{C}$ and $\mathsf{D}$ and the probability that $\mathsf{B}$ beats $\mathsf{C}$ and $\mathsf{D}$. This puts us within a four-dimensional space, and exact orthant probabilities are known in certain special cases (including this one) in four dimensions \cite{DavidMallows61}. This allows Gehrlein and Fishburn to compute an exact expression for the probability that $\mathsf{A}$ beats $\mathsf{C}$ and $\mathsf{D}$ and that $\mathsf{B}$ beats $\mathsf{C}$ and $\mathsf{D}$. The are four different majority graphs satisfying this (depending on whether $\mathsf{A}$ beats or loses to $\mathsf{B}$, and $\mathsf{C}$ with $\mathsf{D}$), each of which is transitive, and each of these is equally likely, so we can compute the probability 0.7395 of obtaining a transitive graph as the outcome of a random election.

From these values we can compute the value of $\Pr(T_1)$ and the value of $\Pr(T_1) + \Pr(T_2)$. Using the fact that $T_2$ and $T_3$ are dual and so $\Pr(T_2) = \Pr(T_3)$, we can compute all of these values. They are displayed in the table below. The \textit{score sequence} is the sequence of out-degrees of the nodes of the tournament. The number of labelings is the number of labeled tournaments with the isomorphism type. The labeled probability is the probability of obtaining $T_i$ specifically as the outcome of an election, while the probability is the probability of obtaining the isomorphism type of $T_i$.

\medskip{}

\noindent \begin{tabular}{c|c|c|c|S[table-format=1.7,group-digits = false]|S[table-format=1.5,group-digits = false]}
	Label	& Score sequence & Linearity & Num Labelings & {Labeled Probability} & {Probability} \\ \hline
	$T_1$ & 3,2,1,0 & 14 & 24 & 0.030813 & 0.7395 \\
	$T_2$ & 3,1,1,1 & 12 & 8 & 0.010628 & 0.08502 \\
	$T_3$ & 2,2,2,0 & 12 & 8 & 0.010628 & 0.08502 \\
	$T_4$ & 2,2,1,1 & 10 & 24 & 0.0037692 & 0.09046 \\
\end{tabular}

%\noindent \begin{tabular}{c|c|c|c|c|c}
%	Label	& Score sequence & Linearity & Num Labelings & Labeled Probability & Probability \\ \hline
%	$T_1$ & 3,2,1,0 & 14 & 24 & 0.0308125 & 0.7395 \\
%	$T_2$ & 3,1,1,1 & 12 & 8 & 0.0106275 & 0.08502 \\
%	$T_3$ & 2,2,2,0 & 12 & 8 & 0.0106275 & 0.08502 \\
%	$T_4$ & 2,2,1,1 & 10 & 24 & 0.003769167 & 0.09046 \\
%\end{tabular}

\medskip{}

\noindent These probabilities agree with our Conjecture \ref{conj}. Note though that we have an interesting phenomenon that while each individual isomorphic copy of $T_4$ is less probable than $T_3$, because there are so many more isomorphic copies of $T_4$ the isomorphism type of $T_4$ on the whole is more likely.

\section{Five Candidates}\label{sec:examples}

Consider the case of five candidates $\mathsf{A}$, $\mathsf{B}$, $\mathsf{C}$, $\mathsf{D}$, and $\mathsf{E}$. As always we fix a standard orientation of the edges between candidates, corresponding to the linear order $\mathsf{A} > \mathsf{B}  > \mathsf{C} > \mathsf{D} > \mathsf{E}$:
\[ \xymatrix{
	& \mathsf{A}\ar[dl]\ar[ddl]\ar[dr]\ar[ddr] \\
	\mathsf{B}\ar[drr]\ar[rr]\ar[d] && \mathsf{E} \\
	\mathsf{C}\ar[rr]\ar[urr] && \mathsf{D}\ar[u]
}\]
From Theorem \ref{lem:cov-calc} we get the following covariance matrix:
\setcounter{MaxMatrixCols}{20}
\[ \bfSigma = \begin{blockarray}{ccccccccccc}
& \mathsf{AB} & \mathsf{AC} & \mathsf{AD} & \mathsf{AE} & \mathsf{BC} & \mathsf{BD} & \mathsf{BE} & \mathsf{CD} & \mathsf{CE} & \mathsf{DE} \\
\begin{block}{c[cccccccccc]}
\mathsf{AB} & \phantom{-} 1 & \phantom{-} 1/3 & \phantom{-} 1/3 & \phantom{-} 1/3 & -1/3 & -1/3 & -1/3 & \phantom{-} 0 & \phantom{-} 0 & \phantom{-} 0 \\
\mathsf{AC} & \phantom{-} 1/3 & \phantom{-} 1 & \phantom{-} 1/3 & \phantom{-} 1/3 & \phantom{-} 1/3 & \phantom{-} 0 & \phantom{-} 0 & -1/3 & -1/3 & \phantom{-} 0 \\
\mathsf{AD} & \phantom{-} 1/3 & \phantom{-} 1/3 & \phantom{-} 1 & \phantom{-} 1/3 & \phantom{-} 0 & \phantom{-} 1/3 & \phantom{-} 0 & \phantom{-} 1/3 & \phantom{-} 0 & -1/3 \\
\mathsf{AE} & \phantom{-} 1/3 & \phantom{-} 1/3 & \phantom{-} 1/3 & \phantom{-} 1 & \phantom{-} 0 & \phantom{-} 0 & \phantom{-} 1/3 & \phantom{-} 0 & \phantom{-} 1/3 & \phantom{-} 1/3 \\
\mathsf{BC} & -1/3 & \phantom{-} 1/3 & \phantom{-} 0 & \phantom{-} 0 & \phantom{-} 1 & \phantom{-} 1/3 & \phantom{-} 1/3 & \phantom{-}- 1/3 & -1/3 & \phantom{-} 0 \\
\mathsf{BD} & -1/3 & \phantom{-} 0 & \phantom{-} 1/3 & \phantom{-} 0 & \phantom{-} 1/3 & \phantom{-} 1 & \phantom{-} 1/3 & \phantom{-} 1/3 & \phantom{-} 0 & -1/3 \\
\mathsf{BE} & -1/3 & \phantom{-} 0 & \phantom{-} 0 & \phantom{-} 1/3 & \phantom{-} 1/3 & \phantom{-} 1/3 & \phantom{-} 1 & \phantom{-} 0 & \phantom{-} 1/3 & \phantom{-} 1/3 \\
\mathsf{CD} & \phantom{-} 0 & -1/3 & \phantom{-} 1/3 & \phantom{-} 0 & -1/3 & \phantom{-} 1/3 & \phantom{-} 0 & \phantom{-} 1 & \phantom{-} 1/3 & -1/3 \\
\mathsf{CE} & \phantom{-} 0 & -1/3 & \phantom{-} 0 & \phantom{-} 1/3 & -1/3 & \phantom{-} 0 & \phantom{-} 1/3 & \phantom{-} 1/3 & \phantom{-} 1 & \phantom{-} 1/3 \\
\mathsf{DE} & \phantom{-} 0 & \phantom{-} 0 & -1/3 & \phantom{-} 1/3 & \phantom{-}  0 & -1/3 & \phantom{-} 1/3 & -1/3 & \phantom{-} 1/3 & \phantom{-} 1\\
\end{block}
\end{blockarray}.\]
Up to isomorphism, there are twelve different tournaments on the five candidates. Recall that the \textit{score sequence} of a tournament is the sequence of outdegrees. These twelve isomorphism types have between them nine different score sequences; one score sequence, $3,3,2,1,1$, has two different isomorphism types, and another, $3,2,2,2,1$, has three different isomorphism types. We number these tournaments from 1 to 12.

\[ \xymatrix{
	{}\save[]+<0.5cm,-0.6cm>*\txt<8pc>{$T_1$:} \restore && \mathsf{A}\ar[dl]\ar[ddl]\ar[dr]\ar[ddr] &&{}\save[]+<0.5cm,-0.6cm>*\txt<8pc>{$T_2$:} \restore && \mathsf{A}\ar[dl]\ar[ddl]\ar[dr]\ar[ddr] &&{}\save[]+<0.5cm,-0.6cm>*\txt<8pc>{$T_3$:} \restore && \mathsf{A}\ar[dl]\ar[ddl]\ar[dr]\ar[ddr]\\
	&\mathsf{B}\ar[drr]\ar[rr]\ar[d] && \mathsf{E} && \mathsf{B}\ar[rr]\ar[d] && \mathsf{E} &&
	\mathsf{B}\ar[drr]\ar[rr]\ar[d] && \mathsf{E}\ar[dll] \\
	&\mathsf{C}\ar[rr]\ar[urr] && \mathsf{D}\ar[u] && \mathsf{C}\ar[rr]\ar[urr] && \mathsf{D}\ar[ull] \ar[u] &&
	\mathsf{C}\ar[rr] && \mathsf{D} \ar[u]}\]

\[ \xymatrix{
	{}\save[]+<0.5cm,-0.6cm>*\txt<8pc>{$T_4$:} \restore && \mathsf{A}\ar[dl]\ar[dr]\ar[ddr]
	&&{}\save[]+<0.5cm,-0.6cm>*\txt<8pc>{$T_5$:} \restore &&
	\mathsf{A}\ar[dl]\ar[ddl]\ar[dr]\ar[ddr]
	&&{}\save[]+<0.5cm,-0.6cm>*\txt<8pc>{$T_6$:} \restore &&
	\mathsf{A}\ar[dl]\ar[ddl]\ar[dr] \\
	&\mathsf{B}\ar[drr]\ar[rr]\ar[d] && \mathsf{E} &&
	\mathsf{B}\ar[drr]\ar[d] && \mathsf{E}\ar[ll]\ar[dll] &&
	\mathsf{B}\ar[drr]\ar[rr]\ar[d] && \mathsf{E} \\
	&\mathsf{C}\ar[rr]\ar[urr]\ar[uur] && \mathsf{D}\ar[u] &&
	\mathsf{C}\ar[rr] && \mathsf{D}\ar[u] &&
	\mathsf{C}\ar[rr]\ar[urr] && \mathsf{D}\ar[u]\ar[uul]\\
}\]

\[ \xymatrix{
	{}\save[]+<0.5cm,-0.6cm>*\txt<8pc>{$T_8$:} \restore && \mathsf{A}\ar[dl]\ar[ddl]\ar[ddr]
	&&{}\save[]+<0.5cm,-0.6cm>*\txt<8pc>{$T_7$:} \restore && \mathsf{A}\ar[dl]\ar[ddl]\ar[ddr] 
	&&{}\save[]+<0.5cm,-0.6cm>*\txt<8pc>{$T_{9}$:} \restore &&
	\mathsf{A}\ar[dl]\ar[ddl]\ar[ddr]\\
	&\mathsf{B}\ar[drr]\ar[rr]\ar[d] && \mathsf{E}\ar[ul]\ar[dll] &&
	\mathsf{B}\ar[drr]\ar[rr]\ar[d] && \mathsf{E}\ar[ul] &&
	\mathsf{B}\ar[drr]\ar[d] && \mathsf{E}\ar[ll]\ar[ul]  \\
	& \mathsf{C}\ar[rr] && \mathsf{D}\ar[u] &&
	\mathsf{C}\ar[rr]\ar[urr] && \mathsf{D}\ar[u] &&
	\mathsf{C}\ar[rr]\ar[urr] && \mathsf{D}\ar[u]}\]

\[ \xymatrix{
	{}\save[]+<0.5cm,-0.6cm>*\txt<8pc>{$T_{10}$:} \restore && \mathsf{A}\ar[dl]\ar[ddr]
	&&{}\save[]+<0.5cm,-0.6cm>*\txt<8pc>{$T_{11}$:} \restore && 
	\mathsf{A}\ar[dl]\ar[ddl]\ar[ddr]
	&&{}\save[]+<0.5cm,-0.6cm>*\txt<8pc>{$T_{12}$:} \restore &&
	\mathsf{A}\ar[dl]\ar[ddl] \\
	& \mathsf{B}\ar[drr]\ar[d] && \mathsf{E}\ar[ll]\ar[ul] &&
	\mathsf{B}\ar[rr]\ar[d] && \mathsf{E}\ar[ul] &&
	\mathsf{B}\ar[drr]\ar[d] && \mathsf{E}\ar[ul]\ar[ll] \\
	& \mathsf{C}\ar[rr]\ar[urr]\ar[uur] && \mathsf{D}\ar[u] &&
	\mathsf{C}\ar[rr]\ar[urr] && \mathsf{D}\ar[u]\ar[ull] &&
	\mathsf{C}\ar[rr]\ar[urr] && \mathsf{D}\ar[u]\ar[uul] \\
}\]

In the table below we list: their score sequence; linearity; the number of different labeled tournaments having that isomorphism type; the probability of obtaining that isomorphism type as the majority graph of a random profile under IC; and the probability of obtaining a fixed labeling of that isomorphism type. There are no known closed form solutions for orthant probabilities in such high dimensions, so the probabilities given are numeric approximations computed using the R \cite{R} package \texttt{orthant} \cite{orthant}, which contains an implementation of Craig's algorithm from \cite{Craig}.

\medskip{}

\sisetup{
	table-format=1.5,
	group-digits = false,
	input-symbols = {--}
}

\begin{tabular}{c|c|c|c|S[table-format=1.7,group-digits = false]|S[]}
Label	& Score sequence & Linearity & Num Labelings & {Labeled Probability} &  {Probability} \\ \hline
	$T_1$ & 4,3,2,1,0 & 30 & 120 & 0.00439 & 0.527 \\
	$T_2$ & 4,2,2,2,0 & 28 & 40 & 0.00177 & 0.0708 \\
	$T_3$ & 4,3,1,1,1 & 28 & 40 & 0.00169 & 0.0677 \\
	$T_4$ & 3,3,3,1,0 & 28 & 40 & 0.00169 & 0.0677 \\
	$T_5$ & 4,2,2,1,1 & 26 & 120 & 0.000695 & 0.0834 \\
	$T_6$ & 3,3,2,2,0 & 26 & 120 & 0.000695 & 0.0834\\
	$T_7$ & 3,3,2,1,1 & 24 & 120 & 0.000274 & 0.0329\\
	$T_8$ & 3,3,2,1,1 & 24 & 120 & 0.000264 & 0.0317 \\
	$T_{11}$ & 3,2,2,2,1 & 22 & 120 & 0.000123 & 0.0148 \\
	$T_9$ & 3,2,2,2,1 & 22 & 120 & 0.000122 & 0.0147 \\
	$T_{10}$ & 3,2,2,2,1 & 22 & 40 & 0.000118 & 0.00471 \\
	$T_{12}$ & 2,2,2,2,2 & 20 & 24 & 0.0000579 & 0.00139 \\
\end{tabular}

%\begin{tabular}{c|c|c|c|c|c}
%	Label	& Score sequence & Linearity & Num Labelings & Labeled Probability &  Probability \\ \hline
%	$T_1$ & 4,3,2,1,0 & 30 & 120 & 0.004390028 & 0.52680336 \\
%	$T_2$ & 4,2,2,2,0 & 28 & 40 & 0.001769599 & 0.07078396 \\
%	$T_3$ & 4,3,1,1,1 & 28 & 40 & 0.001691763 & 0.06767052 \\
%	$T_4$ & 3,3,3,1,0 & 28 & 40 & 0.001692084 & 0.06768336 \\
%	$T_5$ & 4,2,2,1,1 & 26 & 120 & 0.0006946213 & 0.083354556 \\
%	$T_6$ & 3,3,2,2,0 & 26 & 120 & 0.0006956062 & 0.083472744\\
%	$T_7$ & 3,3,2,1,1 & 24 & 120 & 0.0002639967 & 0.031679604 \\
%	$T_8$ & 3,3,2,1,1 & 24 & 120 & 0.0002745208 & 0.032942496\\
%	$T_9$ & 3,2,2,2,1 & 22 & 120 & 0.0001224644 & 0.014695728 \\
%	$T_{10}$ & 3,2,2,2,1 & 22 & 40 & 0.000117855 & 0.0047142 \\
%	$T_{11}$ & 3,2,2,2,1 & 22 & 120 & 0.0001231687 & 0.014780244\\
%	$T_{12}$ & 2,2,2,2,2 & 20 & 24 & 0.00005787867 & 0.00138908808\\
%\end{tabular}

\medskip{}

\noindent The chance of having a majority winner is the chance of getting $T_1$, $T_2$, $T_3$, or $T_5$; our values give 0.748612396 which closely agrees with 0.74869 from Gehrlein and Fishburn \cite{GehrleinFishburn76}. Gehlrein \cite{Gehrlein1988} used a Monte-Carlo simulation to obtain an estimate of the probability of obtaining a transitive tournament as 0.529; we computed 0.527.

Once again we verify that Conjecture \ref{conj} is true for five candidates. We have an interesting new phenomenon, which is that there are non-dual tournaments with the same linearity. For example, $T_2$, $T_3$, and $T_4$ all have the same linearity, and $T_{3}$ and $T_{4}$ are dual, but $T_2$ is self-dual. We have computed $\Pr(T_2) > \Pr(T_3) = \Pr(T_4)$. ($T_3$ and $T_4$, and $T_5$ and $T_6$, are the only two non-self-dual tournaments here.) If Conjecture \ref{conj} is true, we must then ask what makes one tournament more probable than another tournament with the same linearity?

\section{Monte Carlo simulations}\label{sec:MonteCarlo}

Suppose that we want to perform a Monte Carlo simulation of a margin based voting method, for example to see how often it selects multiple winners. Given a set $\mc{V}$ of $n$ voters, a Monte Carlo simulation might repeatedly generate for each voter a random ballot and compute the margin graph from all of these ballots. The problem is that if we want to do this with a large number of voters, we have to generate a random ballot for each voter and tally them all, and it might take a long time to generate enough random values to get an accurate Monte Carlo simulation.

Using the results of this paper we get a more efficient method. In the notation of Section \ref{sec:profiles}, the random margin graph generated by $n$ voters will be $\sum_{i = 1}^n \bfX^{\mathsf{v}_i}$. By the central limit theorem, $\sqrt{n} \cdot \bfS_n = \frac{1}{\sqrt{n}} \sum_{i = 1}^n \bfX^{\mathsf{v}_i}$ converges in distribution to the multivariate normal distribution $\mc{N}(\mathbf{0},\bfSigma)$ with the same mean and covariance matrix, which has probability density function
\[ f(\mathbf{x}) = \frac{e^{-\frac{1}{2}\mathbf{x}^T\bfSigma^{-1}\mathbf{x}}}{\sqrt{(2\pi)^k|\bfSigma|}}.\]
So instead of generating random profiles, we can instead use a random variable $\bfY \sim \mc{N}(\mathbf{0},\bfSigma)$. We can easily compute from $\bfY$ the majority graph, the qualitative margin graph, and (essentially) the margin graph. This is enough to determine the winning set by any qualitative margin-based voting method, and by most margin-based voting methods such as Borda count.

To generate values of $\bfY$, one can either simply use a software package, or one can use independent normal random variables $\mathbf{Z} = (Z_{i,j})_{i<j}$, and
\[ \bfY = \mathbf{A} \mathbf{Z} \]
where $\mathbf{A}$ is such that $\bfSigma = \mathbf{A} \mathbf{A}^T$. One can use for example $\mathbf{A} = \mathbf{U} \mathbf{\Lambda}^{1/2}$ obtained from the spectral decomposition $\bfSigma = \mathbf{U}\mathbf{\Lambda}\mathbf{U}^{-1}$ of $\bfSigma$, where $\mathbf{\Lambda}$ is the diagonal sequence of eigenvalues and $\mathbf{U}$ is the matrix whose columns are the eigenvectors of $\bfSigma$. Since $\mathbf{\Lambda}$ is diagonal and $\mathbf{U}$ is the matrix of eigenvectors, which are quite sparse, $\mathbf{A}$ is quite sparse, so there is some advantage to generating our own random variables to take advantage of this.

To generate a random profile in this way, one must generate a vector $\mathbf{Z}$ of $\ell(\ell-1)/2 = O(\ell^2)$ independent normal random variables, and the matrix multiplication $\bfY = \mathbf{A} \mathbf{Z}$ is an $\ell(\ell-1)/2 \times \ell(\ell-1)/2$ matrix multiplied by a vector, and is $O(\ell^4)$ (which can be improved by taking advantage of the sparsity of $\mathbf{A}$). Thus rather than simulating some number $n \gg \ell^4$ of voters, we get an improvement using our method.

Sample code for generating random margin graphs is available at \url{https://github.com/MatthewHT/RandomMarginGraphs/}. As a test, we compute some values of interest to Holliday and Pacuit \cite{HollidayPacuit}, namely the size of a winning set using their voting method Split Cycle---see Figure 9 of that paper. %Their largest computation was 10000 profiles with 30 candidates and 5001 voters, which took ? to compute.
We were able, with a day or two of computation time, and a naive algorithm not taking advantage of the sparsity of our matrices, to generate 1,000,000 profiles for 5, 7, 10, and 20 candidates, 100,000 profiles with 30 candidates, and 10,000 profiles for 50 and 70 candidates. (Entries $0.00\%$ in the table below are non-zero entries which were rounded down.)

\sisetup{round-mode = places, round-precision = 2, round-integer-to-decimal}

\begin{center}
	\begin{tabular}{c|S[table-format=2.2,group-digits = false]S[table-format=2.2,group-digits = false]S[table-format=2.2,group-digits = false]S[table-format=2.2,group-digits = false]S[table-format=2.2,group-digits = false]S[table-format=2.2,group-digits = false]S[table-format=2.2,group-digits = false]S[table-format=2.2,group-digits = false]S[table-format=2.2,group-digits = false]S[table-format=2.2,group-digits = false]}
		&  \multicolumn{7}{c}{Size of winning set} \\ \cline{2-10}
		$\ell$ & \multicolumn{1}{c}{1} & \multicolumn{1}{c}{2} & \multicolumn{1}{c}{3} & \multicolumn{1}{c}{4} & \multicolumn{1}{c}{5} & \multicolumn{1}{c}{6} & \multicolumn{1}{c}{7} & \multicolumn{1}{c}{8} & \multicolumn{1}{c}{9} &  {Multiple winners}\\ \hline
		5 & 96.7964\si{\percent} & 3.1456\si{\percent} & 0.0580\si{\percent} & \textemdash & \textemdash & \textemdash & \textemdash & \textemdash & \textemdash & 3.2036\si{\percent} \\
		7 & 92.1850\si{\percent} & 7.3839\si{\percent} & 0.4232\si{\percent} & 0.0079\si{\percent} & \textemdash & \textemdash & \textemdash & \textemdash & \textemdash & 7.8150\si{\percent} \\
		10 & 85.2591\si{\percent} & 13.1304\si{\percent} & 1.5193\si{\percent} & 0.0888\si{\percent} & 0.0024\si{\percent} & \textemdash & \textemdash & \textemdash & \textemdash & 14.7409\si{\percent} \\
		20 & 67.8004 \si{\percent} & 24.2220 \si{\percent} & 6.6463 \si{\percent} & 1.1733 \si{\percent} & 0.1457 \si{\percent} & 0.0112 \si{\percent} & 0.0010 \si{\percent} & 0.0001 & \textemdash & 32.2\si{\percent} \\
		30 & 56.581\si{\percent} & 28.467\si{\percent} & 11.237\si{\percent} & 3.006\si{\percent} & 0.616\si{\percent} & 0.085\si{\percent} & 0.007\si{\percent} & 0.001\si{\percent} & \textemdash &  43.42\si{\percent} \\
		50 & 43.35\si{\percent} & 29.95\si{\percent} & 16.55\si{\percent} & 6.88\si{\percent} & 2.62\si{\percent} & 0.51\si{\percent} & 0.11\si{\percent} & 0.03\si{\percent} & \textemdash & 56.65\si{\percent} \\
		70 & 32.68\si{\percent} & 26.71\si{\percent} & 16.73\si{\percent} & 8.77\si{\percent} & 3.54\si{\percent} & 1.21\si{\percent} & 0.49\si{\percent} & 0.05\si{\percent} & 0.01\si{\percent}  & 67.32\si{\percent}
	\end{tabular}
\end{center}

The methods described here have already been used by Holliday and Pacuit \cite{HollidayPacuitC} to measure the probability, for various voting methods, of violating certain desiderata.

\bibliography{References}
\bibliographystyle{alpha}

\end{document}